\documentclass[envcountsame]{llncs}

\usepackage{amsmath}
\usepackage{amssymb}
\usepackage{xcolor}
\usepackage[T1]{fontenc}
\usepackage{hyperref}
\usepackage[capitalize]{cleveref}
\usepackage[normalem]{ulem}

\usepackage[maxbibnames=10,
backref=true
]{biblatex}

\newif\ifnotes\notestrue

\newcommand{\WEEnc}{\mathsf{WE.Enc}}
\newcommand{\WEDec}{\mathsf{WE.Dec}}
\newcommand{\WESec}{\mathsf{WE.Sec}}
\newcommand{\Obf}{\mathsf{Obf}}
\newcommand{\Eval}{\mathsf{Eval}}
\newcommand{\SampYes}{\mathsf{SampYes}}
\newcommand{\SampNo}{\mathsf{SampNo}}
\newcommand{\negl}{\mathrm{negl}}
\newcommand{\QMA}{\mathsf{QMA}}
\newcommand{\QCIP}{\mathsf{QCIP}}
\newcommand{\QPT}{\mathsf{QPT}}
\newcommand{\bL}{\mathbf{L}}
\newcommand{\bLb}{\overline{\mathbf{L}}}
\newcommand{\DecYes}{\mathsf{DecYes}}
\newcommand{\DecNo}{\mathsf{DecNo}}
\newcommand{\ct}{\mathit{ct}}
\newcommand{\Bin}{\mathrm{Bin}}
\newcommand{\val}{\mathrm{val}}
\newcommand{\Ex}{\mathop{\mathbb{E}}}
\DeclareMathOperator*{\argmax}{arg\,max}
\newcommand{\Ideal}{\mathsf{Ideal}}

\newcommand{\Gap}{\textsf{QMA}\text{--}\textsf{QCIP}[2]}

\newcommand{\mor}[1]{{\color{red}[\textbf{Tomoyuki}: #1]}}
\newcommand{\moh}[1]{{\color{blue}[\textbf{Mohammed}: #1]}}

\begin{document}

\title{Separating Quantum Indistinguishability Obfuscation from Falsifiable Assumptions}
\author{Mohammed Barhoush\inst{1} \and Tomoyuki Morimae\inst{2} \and Ramis Movassagh\inst{1}}
\institute{IonQ Inc.\and Kyoto University}
\maketitle
\begin{abstract}
Quantum indistinguishability obfuscation (\textsf{qIO}) 
aims to make a quantum circuit unintelligible
while preserving its functionality.
It serves as a foundational primitive for advanced applications, such as witness encryption (\textsf{WE}) for $\QMA$, 
non-interactive zero-knowledge arguments for
$\QMA$, and attribute-based encryption for \textsf{BQP}. Despite its importance, constructing \textsf{qIO} from standard assumptions remains a major open problem.

In this work, we prove that the security of \textsf{WE} for \textsf{QMA} cannot be based on any falsifiable cryptographic assumption via a restricted class of quantum black-box reductions. Because \textsf{qIO} for null quantum circuits implies \textsf{WE} for \textsf{QMA}, this also separates null--\textsf{qIO} from falsifiable assumptions. Since almost all standard cryptographic assumptions are falsifiable, our result presents a barrier to basing 
\textsf{qIO} on standard cryptographic assumptions.

The reductions we rule out are restricted: the reduction must query the adversary classically, non-adaptively, at the same security parameter, and only on honestly generated ciphertexts. Moreover, our impossibility applies only to \textsf{WE} with classical ciphertexts, and therefore does not rule out \textsf{qIO} with obfuscators whose output is a quantum state. Ruling out more general reductions, as well as more general forms of \textsf{WE} and \textsf{qIO}, remains open.% On the other hand, unlike previous separations, our result rules out reductions even to falsifiable assumptions that are secure only against uniform adversaries.

Our impossibility relies on the existence of a $\Gap$ gap problem, an average-case assumption postulating a \textsf{QMA} language that cannot be verified with two messages of classical communication. 

\end{abstract}

%\tableofcontents

\newpage
\section{Introduction}\label{sec:intro}
\subsection{Background}

Program obfuscation aims to make code unintelligible while preserving its functionality. The strongest formalization, virtual black-box (\textsf{VBB}) obfuscation, is famously impossible for general circuits~\cite{BGI01}. This impossibility led Barak et al.~\cite{BGI01} to propose the weaker notion of \emph{indistinguishability obfuscation} (\textsf{IO}), which only requires that obfuscations of two functionally equivalent circuits be computationally indistinguishable. Despite this seemingly modest guarantee, classical \textsf{IO} has proven to be a central hub of modern cryptography~\cite{GGHRSW13,SW14}. A long line of work recently culminated in a construction of classical \textsf{IO} from well-founded assumptions~\cite{JLS21}. However, 
one of the underlying assumptions (on bilinear maps) is broken by quantum attacks, and no post-quantum instantiation from well-founded assumptions is known.
 
As quantum computers mature, it is natural to ask for obfuscation of \emph{quantum} programs. Beyond simply protecting quantum software, quantum \emph{indistinguishability obfuscation (\textsf{qIO})} acts as a master primitive for quantum-native cryptography. Even in its weakest useful form---obfuscation of \emph{null} quantum circuits, which reject every input---Bartusek and Malavolta~\cite{BM21} demonstrated that \textsf{qIO} implies witness 
encryption (\textsf{WE}) for $\QMA$, a publicly verifiable non-interactive zero-knowledge argument system for $\QMA$, and attribute-based encryption for $\mathsf{BQP}$. 

 In the $\QMA$ setting, the statement is classical (e.g., ``this local Hamiltonian has a low-energy state''), and the witness is a quantum state. Witness encryption---first introduced for \textsf{NP} by Garg et al.~\cite{GGSW13}---for $\QMA$ is particularly relevant to this work so we give an informal definition here. 

\begin{definition}[Witness Encryption for $\QMA$ (Informal)]
Let $L\in\QMA$. A \textsf{WE}
scheme $(\WEEnc,\WEDec)$ for $L$ consists of the following quantum polynomial time algorithms:
\begin{itemize}
  \item $\WEEnc(1^n,x,m)$: {Takes as input a statement $x$, security parameter $n$, and message $m\in \{0,1\}$. It then outputs a classical ciphertext $\ct$. }
  \item $\WEDec(x,\ct,|w_x\rangle)$: Takes as input a statement $x$, a classical ciphertext $\ct$, and a quantum witness $|w_x\rangle$ of $x$. It then outputs a message $\widetilde{m}\in \{0,1\}$. 
\end{itemize}
\end{definition}

\textsf{WE} allows anyone to encrypt a message by using a statement $x\in L$ such that any witness for $x$ decrypts it, while security requires that encryptions using a false statement $x\notin L$ computationally hide the message. The contrast with classical \textsf{NP} is instructive: \textsf{WE} for \textsf{NP} follows from classical \textsf{IO}~\cite{GGSW13,SW14} and hence, via~\cite{JLS21}, from falsifiable assumptions. A falsifiable assumption~\cite{Naor03,GW11} 
is one that can be modeled as an interactive game between an efficient challenger and an adversary, where the challenger efficiently decides whether the adversary won (\cref{def:falsifiable}). The vast majority of cryptographic assumptions---one-way functions, \textsf{LWE}, \textsf{DDH}, and so on---are falsifiable. Whether \textsf{WE} for $\QMA$ admits any foundation 
in falsifiable assumptions has remained completely open.
%Furthermore, direct lattice-based constructions for \textsf{WE} for \textsf{NP} are known from evasive \textsf{LWE}~\cite{VWW22,Tsa22}.

More generally, the state of \textsf{qIO} constructions is unsettled, primarily because they have yet to be realized in the plain model
with a security proof from standard assumptions.
For example, the null--\textsf{qIO} scheme of~\cite{BM21} is proven secure relative to a classical oracle.
Other \textsf{qIO} variants face the same barrier: obfuscators for pseudo-deterministic quantum circuits~\cite{BKNY23,BBV24} 
and for unitary programs \cite{HT25} are likewise proven secure only within an idealized classical-oracle model. Most recently, \cite{EPRINT:BGMS26}
constructed classical obfuscation of pseudo-deterministic quantum
circuits, yet their security proof is again given only relative to a classical oracle.
On the negative side, while quantum \textsf{VBB} obfuscation is known to be impossible~\cite{AF16,ABDS21}, these impossibility results do not apply to \textsf{qIO}.

In short, after a decade of sustained effort, we neither have a construction of \textsf{qIO} in the plain model with a security
proof from standard assumptions, nor an explanation for its absence. This state of affairs motivates the central question of this work:

\begin{center}
\emph{Can we build \textsf{qIO} (or even just \textsf{WE} for $\QMA$) from standard assumptions?}
\end{center}

\subsection{Our Results}

We provide a negative answer to this question for a natural class of security proofs.
We introduce a new assumption, the existence of a $\Gap$ gap problem, and by assuming it,
we prove that the security of \textsf{WE} for \textsf{QMA} cannot be based on any falsifiable cryptographic assumption via a restricted class of quantum black-box reductions. Since
null--\textsf{qIO} implies \textsf{WE} for \textsf{QMA}, the same separation
applies to null--\textsf{qIO}.

Note that \textsf{WE} and \textsf{qIO} are themselves non-falsifiable: verifying that an adversary has won the security game requires deciding whether its chosen statement lies outside the language, or whether the two quantum circuits it produced are both null, which cannot be done efficiently. 
Consequently, one cannot hope to rule out reductions from all assumptions: any primitive can be trivially based on the (non-falsifiable) assumption that the primitive itself is secure. Following Naor~\cite{Naor03} and Gentry–Wichs~\cite{GW11}, restricting attention to falsifiable assumptions is therefore the standard way to obtain meaningful impossibility results for non-falsifiable primitives~\cite{GW11,BMNY26,ADSS24}.

\if0
\textsf{WE} and \textsf{qIO} themselves are \emph{not} falsifiable: deciding whether an adversary's statement lies outside a given language, or whether two quantum circuits are functionally equivalent, is intractable. Therefore, separating \textsf{WE} or \textsf{qIO} from \emph{non}-falsifiable assumptions is impossible\mor{I do not understand this text.}\moh{better?}, as the primitive itself is an assumption of that kind. Instead, falsifiable assumptions are the standard target for impossibility results of such non-falsifiable primitives\mor{what do you mean by "such primitives"? I do not understand this text.}\moh{better?}, as established in previous literature~\cite{GW11,ADSS24,BMNY26}. 
\fi

The main result of this paper is the following:

\begin{theorem}[Main result, informal]\label{thm:informalmain}
Let $L$ be a $\QMA$ language admitting a $\Gap$ gap problem, and let $\Pi$ be a \textsf{WE} scheme for $L$. Then for every falsifiable assumption, either the assumption is false, or there is no quantum black-box reduction establishing the security of $\Pi$ from it, provided the 
queries made by the reduction algorithm to the adversary
are security-parameter-preserving, non-adaptive, {classical}, and
only honestly generated ciphertexts.
\end{theorem}

Because null--\textsf{qIO} implies \textsf{WE} for $\QMA$~\cite{BM21}, we obtain the following corollary:

\begin{corollary}[Informal]\label{cor:informalqio}
Assume that there exists a $\mathsf{QMA}$ language admitting a $\Gap$
gap problem.
Then for every falsifiable assumption, either the assumption is false, or
there is no quantum black-box reduction establishing the security of any
null--\textsf{qIO} from it, subject
to the same restrictions on the reduction as in \cref{thm:informalmain}.
\end{corollary}

Note that \cref{thm:informalmain} only covers \textsf{WE} with classical ciphertexts.
Hence, \cref{cor:informalqio} does not rule out obfuscators that output
quantum states, such as the obfuscators for pseudo-deterministic quantum
circuits~\cite{BKNY23,BBV24} and for unitary
programs~\cite{HT25}. (See also \cref{sec:limitations}).

Our separation assumes the existence of a $\Gap$ gap problem\footnote{$\QCIP$ is the complexity class for interactive protocols over classical
channels~\cite{BLW,GK25}, and $\QCIP[2]$ is its two-message version.}, defined as follows:

\begin{definition}[$\Gap$ gap problem, informal]
A $\Gap$ gap problem for a language $L\in\QMA$ consists of a distribution over yes-instances that is efficiently samplable together with the (quantum) witnesses, and a distribution over no-instances that may be sampled inefficiently, such that no efficient quantum algorithm can distinguish the two distributions, even when given classical access to an oracle that decides every two-message interactive proof over a classical channel.
\end{definition}

Informally, this assumption is an average-case formulation of the statement that some $\QMA$ languages cannot be verified via two messages of classical communication. This is well supported by the current cryptographic landscape, primarily stemming from the known barriers to achieving two-message classical verification of $\QMA$. Mahadev's breakthrough classical verification of quantum computation protocol~\cite{Mah18} requires four messages. While a long line of subsequent works has refined and optimized various aspects of her original scheme~\cite{GV19, Z22, NZ23, BKM+25, ACGH20, CCY20, B21}, four messages remain the state of the art in the plain model \footnote{Although a non-interactive scheme has been achieved in the random oracle model~\cite{ACGH20}, it does not apply to the plain model.} 

Another circumstantial support
is the recent black-box separation by Barhoush, Morimae, Nishimaki, and Yamakawa~\cite{BMNY26}, which proves that non-interactive classical verification of $\QMA$ (technically, a two-message protocol with a one-message setup) cannot be based on falsifiable assumptions. This provides a negative result reinforcing the conjecture that $\QMA \not\subseteq \QCIP[2]$.\footnote{The result \cite{BMNY26} is for classical verifiers, but the result can be actually extended to the quantum verifier case.} Our notion of $\Gap$ gap problem is directly inspired by the $\QMA$--$\mathsf{QCMA}$ gap problem introduced there; however, the choice to strengthen the assumption by replacing $\mathsf{QCMA}$ with $\QCIP[2]$ in the gap problem is motivated by the interactive nature of the \textsf{WE}  primitive, as we detail in Section~\ref{sec:overview}.

Ultimately, our theorem should be read as a structural barrier: if one believes a falsifiable assumption is true, then no security proof of the stated form can exist to base \textsf{WE} for $\QMA$ on it. The four constraints on the reduction---non-adaptive access, preservation of the security parameter, {classical queries,} and honestly generated ciphertext queries (\cref{def:reduction,def:wered})---are natural restrictions discussed at length in 
\cref{sec:limitations,sec:overview}.

Finally, two features distinguish our separation from prior works in the simulatable-adversary
paradigm~\cite{GW11, ADSS24, BMNY26}. First, ours is the
first separation in this paradigm to apply to an interactive security
game: the previous ones all target a one-shot proof system, in which
the adversary’s attack consists of producing a single message.
Meta-reduction-based impossibilities do address interactive security
notions, but they concern different assumption classes and impose
different structural restrictions on the scheme and on the reduction;
see \cref{sec:related-works}. We hope our techniques open the door to
separating further interactive security notions from all falsifiable
assumptions, where one-shot techniques break down. Second, we rule out reductions even to
falsifiable assumptions that are secure only against \emph{uniform} adversaries.
Previous separations~\cite{GW11, ADSS24, BMNY26} are proven for assumptions formulated against non-uniform adversaries, and thus say nothing about assumptions whose
security is inherently uniform. An example of such an assumption is a keyless collision-resistant hash function, which can only be secure against uniform
adversaries, since a non-uniform adversary can simply hardwire a fixed collision
as advice.

\subsection{Related Works}\label{sec:related-works}

\paragraph{Separations from falsifiable assumptions.} Gentry and
Wichs~\cite{GW11} proved that (adaptively-sound) succinct non-interactive
arguments (\textsf{SNARG}s) admit no black-box reduction to any falsifiable
assumption (under subexponential hardness of some subset-membership problems),
introducing the simulatable-adversary paradigm that underlies all subsequent
work in this line, including ours. Alagic, Dachman-Soled, Shingane, and
Struck~\cite{ADSS24} extended the \textsf{SNARG} separation to \emph{quantum}
black-box reductions. Closest to our work, Barhoush, Morimae, Nishimaki, and
Yamakawa~\cite{BMNY26} separated non-interactive classical verification of
$\QMA$ from all falsifiable assumptions, assuming a $\QMA$--$\mathsf{QCMA}$ gap
problem. Their result concerns a two-message argument system, where the
adversary's attack consists of producing a single classical message. Ours
concerns an encryption scheme, whose security game is interactive with two
classical messages produced by the adversary. This difference drives most of the
technical divergence between the two works (see Section~\ref{sec:overview}).

\paragraph{Meta-reductions.} A complementary technique for ruling out reductions for interactive security notions is the meta-reduction paradigm~\cite{FS10,Pas11}, which turns the reduction itself into an efficient attack on the assumption, typically by
restricting its rewinding behavior. Because
the attack produced by a meta-reduction uses no advice, the technique
also applies to assumptions that are secure only against uniform
adversaries. On the other hand, the known meta-reduction
impossibilities either concern restricted assumption classes, such as
non-interactive assumptions~\cite{FS10}, or rely on structural
properties of the scheme under attack, such as special soundness and
unique witnesses~\cite{Pas11}. The simulatable-adversary
paradigm instead exhibits an inefficient adversary that is
indistinguishable from an efficient simulator under the reduction’s
query access, and it applies to arbitrary falsifiable assumptions. The
two approaches impose different restrictions and apply to different
regimes.

\paragraph{Negative results on witness encryption.} The only prior negative
result targeting witness encryption itself concerns a stronger variant. Garg,
Gentry, Halevi, and Wichs~\cite{GGHW14} show that \emph{extractable} \textsf{WE} for
$\mathsf{NP}$ with auxiliary input is implausible, under a special-purpose
obfuscation conjecture. Their result rules out the primitive (conditionally),
whereas ours concerns plain \textsf{WE} for $\QMA$ and rules out a class of
security proofs from falsifiable assumptions. The two results are complementary,
and neither subsumes the other.

\paragraph{Constructions of \textsf{WE} from non-falsifiable assumptions.}
Witness encryption for $\mathsf{NP}$ and null--\textsf{IO} for classical circuits
can be built from evasive \textsf{LWE}~\cite{VWW22,Tsa22}, a plausibly
post-quantum assumption. Evasive \textsf{LWE} is non-falsifiable: its statement
quantifies over auxiliary-input samplers and admits no efficient challenger. 

\paragraph{Black-box separations in quantum cryptography.} A growing body of work
maps the black-box landscape of quantum primitives: oracle separations between
pseudorandom unitaries and one-way functions~\cite{Kre21}, quantum extensions of
classical separations~\cite{HY20,ACC22,CLM23}, and separations among quantum
primitives~\cite{CM24,CCS24}. These results separate primitives from primitives. In
contrast, we separate a primitive from an entire
class of falsifiable assumptions.

\paragraph{Impossibility of black-box obfuscation.}
\textsf{VBB} obfuscation is known to be impossible for general classical
circuits against classical adversaries~\cite{BGI01}. Alagic and
Fefferman~\cite{AF16} initiated the study of quantum obfuscation and
proved, assuming quantum-secure one-way functions, that quantum
circuits cannot be \textsf{VBB} obfuscated into quantum states when the
adversary may hold multiple copies of the obfuscated state.
Subsequently, Alagic, Brakerski, Dulek, and Schaffner~\cite{ABDS21}
showed that, under the quantum hardness of LWE, even classical
circuits cannot be \textsf{VBB} obfuscated into quantum states, where a single
copy of the obfuscated state suffices for the impossibility. Negative
results are also known for indistinguishability obfuscation of
classical circuits. Mahmoody, Mohammed, Nematihaji, Pass, and Shelat~\cite{MMN+16} rule out fully black-box constructions of \textsf{IO}
from one-way functions and from collision-resistant hashing, assuming
one-way functions exist and $\mathsf{NP} \not\subseteq \mathsf{coAM}$,
and Goldwasser and Rothblum~\cite{GR07} show that
statistically secure \textsf{IO} for all circuits would collapse the polynomial
hierarchy. These results separate \textsf{IO} from specific primitives or rule
out statistical security altogether, whereas we separate quantum \textsf{IO}
from an entire class of assumptions.

\subsection{Limitations and Open Problems}\label{sec:limitations}

\paragraph{Restrictions on the reduction.} Our impossibility result applies to
reductions that (1) preserve the security parameter, (2) access the adversary
non-adaptively and classically, and (3) query the adversary only on honestly
generated ciphertexts. We call such reductions \emph{restricted
\textsf{WE}--reductions} (Section~\ref{sec:bbred}). Restriction (3) reflects the
intuition that a valid attack on \textsf{WE} is the ability to break honestly
generated ciphertexts i.e. a ciphertext that was generated by running the honest encryption algorithm on a uniformly random bit. Hence,  a reduction ought to be able to leverage this
ability alone to attack the falsifiable assumption. However, removing any restriction is the main open problem
left by this work. As a
first step toward a full separation of \textsf{WE} for $\QMA$ from falsifiable
assumptions, we believe these restrictions are reasonable.

Prior works~\cite{GW11,ADSS24,BMNY26} can lift the parameter-preserving,
classical-access, and non-adaptivity restrictions by hardwiring the adversary's
behavior as non-uniform advice. As we explain in Section~\ref{sec:overview}, this
route is unavailable in our interactive setting. In exchange, our proof
dispenses with advice altogether: it rules out reductions even to falsifiable
assumptions whose security holds only against uniform adversaries, which is
unclear for the prior separations. We consider this a reasonable trade-off, but
lifting either restriction would be interesting.

There is a further consequence of our security formulation. Since the adversary produces the instance itself,
the reduction can never query the adversary on an instance of its own choosing.
As in previous separations in this line of work~\cite{GW11,ADSS24,BMNY26}, our separation thus concerns an adaptive-style security
notion, and does not address reductions that establish selective, per-instance
security.

\paragraph{\textsf{IO} for classical circuits.} Our techniques do not, and
provably cannot, extend to \textsf{IO} for classical circuits. The reason is
structural: \textsf{IO} for classical circuits implies \textsf{WE} for
$\mathsf{NP}$, and there is no analog of the $\Gap$ gap problem for
$\mathsf{NP}$---the classical witness is a one-message certificate, so a $\QCIP[2]$
oracle decides the language outright. In fact, classical \textsf{IO} (without quantum security) can be constructed
from falsifiable assumptions~\cite{JLS21}, and hence so can
\textsf{WE} for \textsf{NP}, so no separation can exist for the classical notion unless
one of the underlying assumptions is false.

A stronger notion is post-quantum \textsf{IO}, namely \textsf{IO} for classical
circuits that is quantumly secure. A separation for post-quantum null--\textsf{IO}
would face the same structural barrier, since null--\textsf{IO} for classical
circuits gives \textsf{WE} for $\mathsf{NP}$. Here, interestingly, no construction from
well-studied falsifiable assumptions is known either. The
known plausibly post-quantum constructions of \textsf{WE} for $\mathsf{NP}$ and
null--\textsf{IO} rely on evasive \textsf{LWE}~\cite{VWW22,Tsa22}, a non-falsifiable
assumption, or on new \textsf{LWE}-with-hints assumptions
introduced for
\textsf{IO}~\cite{BDGM20,GP21,WW21,HJL21},
which are falsifiable but ad hoc, and several of which have been
broken~\cite{HJL21}. Post-quantum null--\textsf{IO} for classical
circuits is thus a curious intermediate case: neither a construction
from standard assumptions nor a black-box separation is known, and our
technique provably cannot supply the latter.

The status of post-quantum \textsf{IO} for \emph{general} classical circuits is
open in both directions: neither a construction from standard assumptions nor a
black-box separation is currently known.
This is a particularly important direction for future research, as
general-circuit post-quantum \textsf{IO} serves as the primary known route to
primitives such as public-key quantum money and quantum software
copy-protection.  While our structural barrier implies that
proving such a separation would require techniques different from ours, we hope that the ideas in this work help in this direction. 

\paragraph{\textsf{qIO} with quantum output.} Notably, we only consider
obfuscators that output classical bit strings. Consequently, our impossibility
leaves untouched any obfuscator whose output is a \emph{quantum} state. This gap
is not vacuous, as schemes of precisely this excluded kind exist. The obfuscators
of Bartusek, Kitagawa, Nishimaki, and Yamakawa~\cite{BKNY23} and of Bartusek,
Brakerski, and Vaikuntanathan~\cite{BBV24} act on pseudo-deterministic quantum
circuits and yield obfuscated programs encoded as quantum states, while Huang and
Tang~\cite{HT25} obfuscate unitary programs, again into quantum states. Because
each of these produces an obfuscation with quantum output, all three fall beyond
the syntactic class considered here, and our result carries no implication for
them.

This exclusion is structural, not merely an artifact of our proof. With quantum
ciphertexts, we would need an oracle for two-message protocols over a \emph{quantum}
channel, that is, a $\mathsf{QIP}[2]$-type oracle. Since a quantum witness is
itself a one-message quantum certificate ($\QMA=\mathsf{QIP}[1]\subseteq
\mathsf{QIP}[2]$), such an oracle decides the language outright, so no gap problem
relative to it can exist. This mirrors the barrier for $\mathsf{NP}$ above.

\paragraph{Instantiating the assumption.} Our separation assumes the existence of
a $\Gap$ gap problem. It would be valuable to base this gap problem on
concrete quantum hardness assumptions, or at least to give an oracle separation
between $\QMA$ and $\QCIP[2]$. %However, note that our assumption is an average-case hardness rather than worst-case hardness, so this would only provide some motivation for the existence of this problem. 

The corresponding one-message question was recently settled: Bostanci,
Haferkamp, Nirkhe, and Zhandry~\cite{BHNZ25} give a classical oracle separating
$\QMA$ from $\mathsf{QCMA}$. This makes the oracle separation of $\QMA$ from
$\QCIP[2]$ the concrete next step, and it would be interesting to see whether their
oracle already defeats two-message classical-channel protocols. Exploring this
gap more concretely could weaken our assumption and strengthen the separation,
and would be of independent interest for quantum complexity theory.

\section{Technical Overview}\label{sec:overview}

\subsection{The Simulatable-Adversary Paradigm}\label{sec:paradigm}
Our impossibility result follows the simulatable-adversary paradigm of Gentry
and Wichs~\cite{GW11}. To rule out black-box reductions, one exhibits two
adversaries against the primitive: a computationally \emph{unbounded} adversary
$\overline{P}$ that genuinely breaks it, and an \emph{efficient} simulator $P$
that no efficient algorithm can distinguish from $\overline{P}$ under the query
access the reduction is permitted.

Why does such a pair suffice? Suppose a reduction $\Sigma$ bases the security of \textsf{WE} on a falsifiable assumption represented by a $\QPT$ challenger $C$. Informally, a falsifiable assumption is defined by an efficient game in which no efficient adversary should be able to make $C$ output $1$ with noticeable advantage. A black-box reduction $\Sigma$ is an algorithm that, given oracle access to an adversary breaking \textsf{WE}, can break the falsifiable assumption. Therefore, if $\overline{P}$ is a computationally unbounded adversary that breaks \textsf{WE}, then $C$ outputs $1$ with non-negligible advantage when interacting with $\Sigma^{\overline{P}}$. But $\Sigma $ is efficient, so replacing its oracle access $\overline{P}$ by
$P$ shifts its winning probability by at most a negligible amount. In other words, $\Sigma^{P}$ still wins the falsifiable assumption game with non-negligible advantage, and it is efficient, since $P$ is efficient. Therefore, this
shows that the assumption was already false since it can be broken with an efficient attack. In summary, either there does not exist such a reduction or the assumption is false. The entire separation therefore
reduces to two tasks: constructing the pair $(\overline{P},P)$, and proving that
no permitted  reduction tells them apart.

However, there are multiple challenges that come up when we try to apply the simulatable-adversary paradigm to our setting.

\subsection{Technical Challenges}

\paragraph{Challenge 1: \textsf{WE} security game is not appropriate.}
Before applying the simulatable-adversary paradigm, we first need to modify the \textsf{WE} security game. This is because the textbook definition of \textsf{WE} does not have a fundamental property of previous primitives that were separated using the simulatable-adversary paradigm. In particular, the adversary in previous separations \cite{GW11,BMNY26} gets to \emph{choose} the instance $x$ as part of the attack, which turns out to be crucial. We formulate a security notion for \textsf{WE} with this property. Our \textsf{WE} security game (Definition~\ref{def:weSec})
is \emph{interactive}: rather
than fixing an instance $x \notin L$ in advance and asking the adversary to
distinguish encryptions of $0$ and $1$ under it, we let the adversary
\emph{choose} its own instance. It outputs an instance $x$, receives a ciphertext encrypted under $x$, and wins if it
guesses the encrypted bit correctly while ensuring $x \notin L$ (see \cref{sec:we} for more details). Because the
adversary now carries the extra burden of sampling the no-instance itself, this
notion is weaker than the standard one, and ruling out reductions for a weaker
notion yields a \emph{stronger} impossibility result.

%In separations for argument systems such as \textsf{SNARG}s~\cite{GW11} or \textsf{CVQC}~\cite{BMNY26}, the target adversary is a \emph{one-shot} entity: given a statement, it immediately returns a proof. This lets the adversary be modeled as a memoryless map from a query to a response, a structural feature that simulation arguments rely on heavily. Witness encryption does not fit this template.

The price of this design is interactivity: the attack is no longer a static map
but a two-phase process---the adversary first commits to an instance, and only then answers a ciphertext. This structure
breaks the separation techniques of~\cite{GW11,ADSS24,BMNY26}, which lean throughout on
the non-interactive, one-shot nature of their adversaries, as we discuss in the next two challenges.

\paragraph{Challenge 2: Finding $\overline{P}$.}
The paradigm requires an unbounded adversary $\overline{P}$ that breaks the scheme while remaining indistinguishable from an efficient $P$. To show this indistinguishability, our strategy is to  \emph{reproduce the behavior of $\overline{P}$ and $P$ using a $\QCIP[2]$ oracle}. 
Here, a $\QCIP[2]$ oracle $O$ is defined as follows:
For $\QPT$ algorithms, $V_1$ and $V_2$, and a string $x$, we define
\begin{equation}
  \val_{V_1,V_2}(x)\ :=\ \max_{\{z_y\}_y}\ \Pr_{(\rho,y)\leftarrow V_1(x)}\big[V_2(x,\rho,y,z_y)=1\big],
\end{equation}
where $\rho$ is a quantum state.
Then a $\QCIP[2]$ oracle $O$ is defined by
\begin{equation}
  O(V_1,V_2,x)=
  \begin{cases}
    1 & \text{if } \val_{V_1,V_2}(x)\ \ge\ 2/3,\\[2pt]
    0 & \text{if } \val_{V_1,V_2}(x)\ \le\ 1/3.
  \end{cases}
\end{equation}
On inputs with $1/3<\val_{V_1,V_2}(x)<2/3$ the output of $O$ is unconstrained.
In other words, a $\QCIP[2]$ oracle can find the optimal prover of the $\QCIP[2]$ protocol specified by the verifier $(V_1,V_2)$, and
decide whether the optimal prover makes the verifier $V_2$ accept with high probability or low probability.

\if0
\begin{definition}[Informal $\QCIP[2]$]
A language $L$ is in $\QCIP[2]$ if
there exist $\QPT$ algorithms $V_1,V_2$ such that:
\begin{itemize}
  \item \emph{Completeness:} for every $x\in L$, there is a family $\{z_y\}$ of
        $\mathrm{poly}(|x|)$-length strings with
        \begin{equation}
          \Pr_{(\rho,y)\leftarrow V_1(x)}[V_2(x,\rho,y,z_y)=1]\ \ge\ 2/3.
        \end{equation}
  \item \emph{Soundness:} for every $x\notin L$ and every such
        family $\{z_y\}$,
        \begin{equation}
          \Pr_{(\rho,y)\leftarrow V_1(x)}[V_2(x,\rho,y,z_y)=1]\ \le\ 1/3.
        \end{equation}
\end{itemize}
\end{definition}
\fi

The main difficulty is that it is unclear how to use a $\QCIP[2]$ oracle to simulate $P$ or $\overline{P}$. A $\QCIP[2]$ query returns the value achieved by an \emph{optimal} prover depending on the reduction and will not in general act as $P$ or $\overline{P}$. 

Previous works avoid this issue by taking a union bound over all possible prover strategies; see \cite{GW11,BMNY26,ADSS24}. The details of their argument are not important but essentially, in their non-interactive setting, a prover gives only a single polynomial-length answer. Hence, there are at most singly exponentially many strategies. In our setting, the prover must answer every possible ciphertext, so there are $2^{2^{\textsf{poly}(n)}}$ such functions. This space is too large for the same arguments used in previous works.

We are therefore forced to take a different approach: we construct $\overline{P}$ explicitly, requiring it to simultaneously satisfy two conditions:
\begin{enumerate}
  \item it breaks the {\textsf{WE}} scheme; and
  \item it can be simulated by a $\QCIP[2]$ oracle.
\end{enumerate}
While achieving either property independently is straightforward, satisfying both simultaneously is the difficult part.

Our primary candidate is the \emph{Bayesian decryptor}: given a ciphertext, it outputs the most likely plaintext. This satisfies Condition~1 optimally, as no strategy can achieve higher inversion accuracy. 

However, it is not clear if this satisfies Condition~2. At first glance, there is no guarantee that the interaction between the reduction and the Bayesian decryptor can be simulated by a $\QCIP[2]$ oracle. The $\QCIP[2]$ oracle can be used to check for the existence of a prover that makes the {verifier} accept with $2/3$ probability. But this optimal prover has no inherent incentive to invert ciphertexts; it will simply adopt whatever behavior maximizes acceptance for the reduction. 

We resolve this by redesigning the protocol to force the prover's optimal strategy to be inversion. Among the ciphertexts produced by the reduction, we interleave additional \emph{planted ciphertexts} that we encrypt ourselves (and whose plaintexts we know). The verifier is instructed to reject instantly if any planted ciphertext is answered incorrectly. Under this mechanism, a prover cannot risk deviating from inversion: to pass the test, it must invert the planted ciphertexts, and because the planted ciphertexts are indistinguishable from the reduction's ciphertexts (by our assumption of honestly generated ciphertexts), it must invert all of them. Optimal play is thus forced to be Bayesian inversion, ensuring that the $\QCIP[2]$ oracle behaves exactly as $\overline{P}$. This requirement also highlights why the reduction must produce honestly generated ciphertexts: planted ciphertexts can only remain hidden if drawn from the exact same distribution.

\paragraph{Challenge 3: Our restrictions and advice.} 
Our separation is also limited to reductions that query the adversary on the same security parameter, non-adaptively and classically. In general, reductions may query their adversary on security parameters smaller than their own, adaptively, or quantumly. Prior separations~\cite{GW11,ADSS24,BMNY26} handle these situations by a sequence of hybrids that progressively strengthen the
indistinguishability of $P$ and $\overline{P}$: from a single sample, to
polynomially many samples, to adaptive quantum queries, and finally to adaptive
quantum queries across all security parameters. Each of these steps exploits the fact that the adversary's behavior can
be hardwired as non-uniform advice. However, in our setting, the adversary is an interactive map of exponential size, which prevents it from being compressed into non-uniform advice. 

We therefore restrict our reductions in such a way that the queries are
security-parameter preserving, non-adaptive, classical, and honestly-generated ciphertexts.
{Concretely, our restricted reduction $\Sigma$ interacts with its adversary as follows, for some polynomials $s$ and $\ell$ determined by the reduction. First, $\Sigma$ queries the first stage of the adversary $s$ times, receiving instances $x_1, \ldots, x_s$. Then, for each $i \in [s]$, $\ell$ ciphertexts $ct_{i,1}, \ldots, ct_{i,\ell}$ are generated by honestly running the encryption algorithm on $x_i$ and independently sampled uniform plaintext bits. Finally, $\Sigma$ queries the second stage of the adversary on the pairs $(x_i, ct_{i,j})$ and receives one plaintext guess for each. Thus, the queries are security-parameter-preserving, classical, non-adaptive, and made only on honestly generated ciphertexts.}

{However,} there is a compensating benefit for our approach.  As mentioned, the proofs in~\cite{GW11,ADSS24,BMNY26} rely on non-uniform advice across multiple steps, leaving it unclear whether their separations hold in a fully uniform setting. Our separation avoids this dependence, and therefore rules out black-box reductions even from falsifiable assumptions that are only uniformly secure.

% Rewritten versions of ``Barrier 2'' and ``Our Separation''.
% The numerical citations and theorem/lemma numbers follow the supplied draft.

\subsection{Our Separation}\label{sec:oursep}

Our separation is demonstrated in the following steps. 

\paragraph{Decryptability transfers to no-instances.}
Nothing in the definition of {\textsf{WE}} prevents $\WEEnc(1^n,x,\cdot)$ from becoming statistically lossy when $x \notin L$; in fact, lossiness on no-instances would naturally provide unconditional security. Our first main step (Lemma~\ref{lem:decoff}) rules out this possibility by showing 
that{, assuming the existence of a $\Gap$ problem,} ciphertexts generated with no-instances can be (inefficiently) decrypted and the plaintext message recovered.

This is the first instance where we can employ the $\QCIP[2]$ oracle using the following key observation:
\begin{center}
\emph{``My ciphertext was decrypted correctly'' can be expressed as a two-message protocol over a classical channel.}
\end{center}
To see this, consider a verifier that chooses $k$ random bits, encrypts each bit
under an instance $x$, and sends the resulting ciphertexts to a prover. The
prover returns one guess for each ciphertext, and the verifier checks these
guesses against the original bits.

When $x\in L$, correctness of the witness encryption scheme guarantees the
existence of a prover that succeeds with high probability. Therefore, when we query the $\QCIP[2]$ oracle with this verifier encrypting with a yes-instance, the oracle will return 1 since there exists a strategy (namely running the decryption algorithm) that will make the verifier accept.  The $\Gap$ gap assumption tells us that a $\QCIP[2]$ oracle cannot be used to distinguish yes- and no-instances of a language so the same
decryption property must hold for almost all sampled no-instances. Therefore, even when
$x\notin L$, there exists a (possibly inefficient) strategy that decrypts
ciphertexts correctly with high probability.

\paragraph{Choosing $P$ and $\overline{P}$.}
Following the paradigm in Section~\ref{sec:paradigm}, the problem reduces to
defining two adversaries and proving their indistinguishability. Both are
\emph{two-stage} algorithms matching the interface of the security
game: the first stage produces an instance and the second
stage answers a ciphertext. Concretely, we write $\overline{P} = (\overline
P_1, \overline P_2)$ for the unbounded adversary that breaks the witness
encryption scheme $\Pi = (\WEEnc, \WEDec)$, and $P = (P_1, P_2)$ for the
efficient adversary that no reduction can distinguish from it. 
{Let $L$ be a \textsf{QMA} language that admits a $\Gap$ gap problem.}
We define the two algorithms as follows:
\begin{itemize}
  \item \textbf{Algorithm of $\overline{P}$.} The first stage samples an instance
    \emph{outside} the language,
    \[
      x \leftarrow \overline P_1(1^n),\qquad x \notin L,
    \]
    and the second stage is the Bayesian decryptor: it answers a ciphertext $\ct$ with its most likely
    plaintext,
    \[
      \overline P_2(x, \ct) \;=\; \arg\max_{m}\,\Pr[\WEEnc(1^n,x,m)= \ct].
    \]
    Both stages may be computationally inefficient.
  \item \textbf{Algorithm of $P$.} The first stage samples an instance \emph{inside}
    the language together with a quantum witness,
    \[
      (x, |{w_x}\rangle^{\otimes t}) \leftarrow P_1(1^n),\qquad x \in L,
    \]
    and the second stage decrypts a ciphertext $\ct$ using the honest decryption
    algorithm with this witness,
    \[
      P_2(x, \ct, |{w_x}\rangle^{\otimes t}) \;=\; \WEDec(x, \ct, |{w_x}\rangle^{\otimes t}).
    \]
    Both stages are efficient, where the efficiency of $P_1$ is assumed as part of the $\Gap$ gap problem.
\end{itemize}
Crucially, $\overline P_2$ is a well-defined, deterministic function:
we established that ciphertexts generated from
no-instances can still be decrypted correctly. This is what makes $\overline{P}$ a fixed adversary, independent of
the reduction---and it is precisely this independence that lets us estimate its
behavior with a $\QCIP[2]$ oracle, as we describe next.

\paragraph{The Simulator Protocol and Hidden Tests.}
For any algorithm $Q$, let 
\begin{align}
\Gamma^Q :=\langle C,\Sigma^Q\rangle 
\end{align}
denote the algorithm obtained by viewing the interaction between the falsifiable assumption challenger
$C$ and the reduction $\Sigma^Q$ as a single algorithm.
The output of $\Gamma$ is the output of $C$.
We now build a $\QPT$
algorithm $\tilde{\Gamma}$ that, using only a $\QCIP[2]$ oracle,
reproduces $\Gamma$'s output on both sides:  
\begin{align}
  \Big|\Pr\big[\Gamma^{P}(1^n)=1\big]
       -\Pr_{x\leftarrow{L}_n}\!\big[\tilde{\Gamma}^{\QCIP[2]}(x)=1\big]\Big|
    &\ \le\ \epsilon(n), \label{eq:bridge-yes2}\\[2pt]
  \Big|\Pr\big[\Gamma^{\overline{P}}(1^n)=1\big]
       -\Pr_{x\leftarrow\overline{{L}}_n}\!\big[\tilde{\Gamma}^{\QCIP[2]}(x)=1\big]\Big|
    &\ \le\ \epsilon(n). \label{eq:bridge-no2}
\end{align}
Here, $L_n:= L\cap\{0,1\}^n$ and $\overline{L}_n:= \{0,1\}^n\setminus L_n$.
These equations are essentially saying that $\tilde{\Gamma}^{\QCIP[2]}(x)$ simulates $\Gamma^P$ when $x\in L$ and simulates $\Gamma^{\overline{P}}$ when $x\notin L$. This will allow us to later argue that both interactions are indistinguishable using $\Gap$ gap problem. 

Recall that the reduction $\Sigma$ has two stages. In the first stage, it
queries the first adversary to get instances $x_1,\ldots,x_s$. Then
$\Sigma$ 
generates ciphertexts under these instances, and queries the second adversary
to get one plaintext guess for each ciphertext. We simulate
these two stages by a $\QCIP[2]$ oracle.

The algorithm $\widetilde{\Gamma}^{\QCIP[2]}(x)$ queries the $\QCIP[2]$
oracle on $(V_1,V_2,x)$, which is defined as follows:
\begin{enumerate}
  \item $V_1$ runs the interaction between the challenger $C$ and the
  reduction $\Sigma$, using the given instances
  $x=(x_1,\ldots,x_s)$ in place of the instances that $\Sigma$ would normally
  receive from the first stage of its \textsf{WE} adversary. It continues until
  $\Sigma$ has generated all of its ciphertext queries
  $\ct_{i,j}$. We call these the \emph{real} ciphertexts.

  \item For each instance $x_i$, $V_1$ generates additional \emph{test}
  ciphertexts: $k/2$ encryptions of $0$ and $k/2$ encryptions of $1$ for some polynomial $k$.

  \item $V_1$ randomly shuffles the real and test ciphertexts and sends the
  resulting batch to the prover.

  \item The prover returns one plaintext guess for each ciphertext. The
  verifier $V_2$ rejects if any test ciphertext is answered incorrectly.
  Otherwise, it extracts the answers corresponding to the real ciphertexts,
  restores their original order, and gives them to $\Sigma$ as the answers it
  would have received from the second stage of its \textsf{WE} adversary.
  The interaction between $\Sigma$ and $C$ then continues, and $V_2$ outputs
  the final output bit of $C$.
\end{enumerate}

By definition, the $\QCIP[2]$ oracle evaluates this protocol against an
optimal prover, namely one that maximizes the probability that $V_2$ outputs
$1$. We show that the hidden tests force such a prover to behave essentially
like a decryptor. %By the definition of $P$ and $\overline{P}$, this means that on yes-instances, the prover will behave like $P$ and on no-instances, it will behave like $\overline{P}$.

Notice that the prover does not know
which ciphertexts are real and which are tests. This hiding property follows from the restriction on the reduction. Every real
ciphertext is an honest encryption under $x$, and every test ciphertext is
generated in exactly the same way. The random permutation then hides the
positions of the tests. Within the set of ciphertexts encrypting a fixed bit,
the prover cannot identify which ones came from the reduction and which ones
were planted by the verifier. Therefore, deliberately answering a real
ciphertext incorrectly creates a significant risk of answering a hidden test
incorrectly and being rejected.

The hidden tests therefore force an optimal prover to decrypt all ciphertexts.
As a result, $\tilde{\Gamma}^{\QCIP[2]}(x)$ approximates the interaction between
$\Gamma$ and the appropriate decryptor: $P$ when $x\in L$, and
$\overline{P}$ when $x\notin L$. This gives
\cref{eq:bridge-yes2,eq:bridge-no2}.

We now apply the gap assumption to the $\QPT$ distinguisher
$\tilde{\Gamma}$. The assumption states that no efficient algorithm can distinguish elements inside and outside the language, even given a $\QCIP[2]$ oracle. This means that the outputs of $\tilde{\Gamma}$ on the two
instance distributions are computationally indistinguishable:
\begin{align}
  \Big|\Pr_{x\leftarrow\overline{{L}}_n}\!\big[\tilde{\Gamma}^{\QCIP[2]}(x)=1\big]
       -\Pr_{x\leftarrow{L}_n}\!\big[\tilde{\Gamma}^{\QCIP[2]}(x)=1\big]\Big|
    \ \le\ \negl(n). \label{eq:gap-step2}
\end{align}
Chaining \eqref{eq:bridge-yes2}, \eqref{eq:bridge-no2}, and \eqref{eq:gap-step2} through the triangle inequality yields
\begin{align}\label{eq:goal2}
   \Big|\Pr\big[\Gamma^{P}(1^n)=1\big] 
       -\Pr\big[\Gamma^{\overline{P}}(1^n)=1\big]\Big|
    \ \le\ 3\epsilon(n).
\end{align}

\paragraph{Conclusion.}
Suppose a reduction $\Sigma$ bases the security of the \textsf{WE} scheme on a falsifiable assumption. Interacting with $\Sigma^{\overline{P}}$, the challenger in the falsifiable assumption game accepts with noticeable advantage because $\overline{P}$ breaks the \textsf{WE} scheme. By \cref{eq:goal2}, replacing $\overline{P}$ with $P$ alters this acceptance probability by only a negligible amount. 

However, $P$ is efficient. Therefore, $\Sigma^P$ is an efficient attack against
the falsifiable assumption. It follows that the assumption must be false.
Consequently, either no reduction of the stated form exists, or the underlying
falsifiable assumption is false. This proves our main result,
Theorem~\ref{thm:informalmain}.

\section{Preliminaries}

\subsection{Notation}
In this work, we denote sets and algorithms by capital letters such as $X, A$ and distributions
by bold such as $\mathbf{L}, \textbf{E}$. All sets and distributions are assumed to be parameterized by
some integer $n$, typically a security parameter. We write
$x \leftarrow \mathbf{D}$ to say that $x$ is sampled according to the
distribution $\mathbf{D}$. If $X$ is a set, then $x \leftarrow X$ means that
$x$ is chosen uniformly at random from the set. Let
$[n] := \{1, 2, \ldots, n\}$ for every $n \in \mathbb{N}$.
Furthermore, let $\mathrm{negl}(n)$ denote a negligible function, i.e., a
function that is asymptotically smaller than the inverse of any polynomial.
We follow the standard notation of quantum information [36]. We say that an
algorithm $A$ is $\QPT$ if it is quantum and runs in polynomial time relative to
its input size. In this work, when we say a polynomial, we always mean a
positive polynomial. We write $\mathrm{poly}(n)$ for an unspecified
polynomial in $n$. A function $f$ is an \emph{inverse polynomial} if
$f(n) = 1/p(n)$ for some polynomial $p$. For $N \in \mathbb{N}$ and
$q \in [0,1]$, we write $\mathrm{Bin}(N, q)$ for the binomial distribution
with $N$ trials and success probability $q$. An algorithm $A$ is non-uniform
if it is initialized with quantum advice and is said to be uniform if no
advice is given. Adversaries in the security definitions and proofs are
assumed to be uniform unless stated otherwise.
For an algorithm $Q$ whose queries and responses are classical strings, we
write $A^{Q}$ to mean that $A$ has classical query access to $Q$. When
$Q = (Q_1, Q_2)$ is a pair of algorithms, $A^{Q}$ means that $A$ has
classical query access to both $Q_1$ and $Q_2$; the two components may share
an internal, possibly quantum, state, which each invocation may read and
update. In all cases the interface is classical: we only consider classical
oracle access in this work.  % For a unitary $U$, we write $A^{|U\rangle}$ to mean that the algorithm $A$ has quantum query access to $U$.
%In either case, we add $\|$ next to the oracle to denote that the access is non-adaptive, meaning that the algorithm can submit multiple queries but all at once. Furthermore, if $A$ only gets access to the oracle on the same security parameter, then we add an $n$ subscript to the oracle. The black-box reductions considered in this work are restricted to non-adaptive, security-parameter-preserving reductions where the reduction gets oracle access to the attack non-adaptively and on the same security parameter (see Section~\ref{sec:bbred}).

\subsection{Indistinguishability Obfuscation for Null Quantum Circuits}

We recall the notion of indistinguishability obfuscation for null quantum
circuits from~\cite{BM21}, restricting attention to quantum circuits producing
a single classical output bit. Recall, a null quantum circuit is one that on any input, outputs 0 except with negligible probability. 

Note that \textsf{qIO} has been studied and defined for more general quantum
circuit classes~\cite{BKNY23,BBV24,HT25}. We refer the reader to these works
for the general definitions, but in this work we define \textsf{qIO} only for
null quantum circuits. 

\begin{definition}\label{def:nulliO}
An indistinguishability obfuscator for null quantum circuits (null--\textsf{qIO})
consists of the following pair of $\QPT$ algorithms.
\begin{itemize}
  \item $\Obf(1^n,Q)$: Given a security parameter $1^n$ and a classical description of a
        quantum circuit $Q$ with a single classical output bit, this algorithm
        produces a classical description of an obfuscated program $\widetilde{Q}$.
  \item $\Eval(\widetilde{Q},|\psi\rangle)$: Given a classical description of a program $\widetilde{Q}$
        and a quantum state $|\psi\rangle$, this algorithm produces a bit $b$.
\end{itemize}
\end{definition}

The correctness requirement is stated next.

\begin{definition}[Correctness]\label{def:qioCorrect}
A null--qIO $(\mathsf{Obf},\mathsf{Eval})$ satisfies correctness if, for
every negligible function $\nu$ and every polynomial $p$, there exist a
polynomial $k$ and a negligible function $\mu$ such that the following
holds for every family $\{Q_n\}_{n\in\mathbb{N}}$ of quantum circuits
with classical descriptions of size at most $p(n)$, states
$\{|\psi_n\rangle\}_{n\in\mathbb{N}}$, and bits
$\{b_n\}_{n\in\mathbb{N}}$ satisfying
$\Pr[Q_n(|\psi_n\rangle)=b_n]\ge 1-\nu(n)$:
\begin{align}
\Pr\Bigl[\mathsf{Eval}\bigl(\widetilde{Q},
|\psi_n\rangle^{\otimes k(n)}\bigr)=b_n
: \widetilde{Q}\leftarrow \mathsf{Obf}(1^n,Q_n)\Bigr]\ge 1-\mu(n).
\end{align}
\end{definition}

\begin{remark}\label{rem:relaxed}
This correctness notion relaxes the usual obfuscation guarantee along two
dimensions. First, it only applies to inputs on which the circuit's behavior is
decisive --- accepting or rejecting with overwhelming probability. Second, it
allows the evaluator to consume several copies of the input state rather than a
single copy. Relaxing the definition strengthens our separation.
\end{remark}

Finally, we state the security requirement.

\begin{definition}[Security of null--qIO]
\label{def:null-qio-security}
Let $\mathcal{O}=(\mathsf{Obf},\mathsf{Eval})$ be a null--qIO scheme.
We say that $\mathcal{O}$ is secure if the following holds.

Let $\{Q_{0,n}\}_{n\in\mathbb{N}}$ and $\{Q_{1,n}\}_{n\in\mathbb{N}}$ be any two polynomial-size families of quantum circuits having the
same size and the same input and output spaces. Suppose that both
families are null. Namely, suppose that there exists a negligible
function $\nu$ such that, for every $b\in\{0,1\}$, every $n$, and every
input state $\rho$, $\Pr\!\left[Q_{b,n}(\rho)=1\right]
    \leq \nu(n).$

Then, for any non-uniform $\QPT$ algorithm
$D_n$,
\[
\begin{aligned}
\Bigl|
    &\Pr\!\left[
        D_n\!\left(
            \mathsf{Obf}(1^n,Q_{0,n})
        \right)=1
    \right] -
    \Pr\!\left[
        D_n\!\left(
            \mathsf{Obf}(1^n,Q_{1,n})
        \right)=1
    \right]
\Bigr|
\leq \negl(n).
\end{aligned}
\]
\end{definition}

\subsection{Quantum Complexity Classes}

We first recall the definitions of the complexity classes Quantum-Merlin-Arthur
($\QMA$) and Quantum-Classical-Interactive-Proofs ($\QCIP[2]$).

\begin{definition}[$\QMA$]\label{def:qma}
A language $L\subseteq\{0,1\}^*$ is in $\QMA$ if there exists a $\QPT$
algorithm $Q$ and a polynomial $p$ such that for any $x\in\{0,1\}^*$,
\begin{enumerate}
  \item If $x\in L$, then there exists a $p(|x|)$-qubit state $|w_x\rangle$ such that
  \begin{equation}
    \Pr[1\leftarrow Q(x,|w_x\rangle)]\ \ge\ \tfrac{2}{3}.
  \end{equation}
  \item If $x\notin L$, then for any $p(|x|)$-qubit state $|w\rangle$,
  \begin{equation}
    \Pr[1\leftarrow Q(x,|w\rangle)]\ \le\ \tfrac{1}{3}.
  \end{equation}
\end{enumerate}
\end{definition}

Henceforth, for a $\QMA$ language $L$, we let $L_n:=\{0,1\}^n\cap L$ and $\overline{L}_n:=\{0,1\}^n\setminus L$.

\begin{definition}[$\QMA$ Relation]\label{def:qmarel}
Let $L$ be a language in $\QMA$ and let $Q$ be its verification algorithm. A
$\QMA$ relation $R_{L,Q}$ is the set of pairs $(x,|w\rangle)$ satisfying
$\Pr[1\leftarrow Q(x,|w\rangle)]\ge\frac{2}{3}$.
\end{definition}

We recall the complexity class $\QCIP$ as defined in~\cite{BLW,GK25}. $\QCIP$
is the complexity class for interactive protocols over classical
channels. We define it as a promise problem rather than a language, because the
oracles of \cref{def:qcipOracle} must be allowed to answer
arbitrarily on ambiguous instances, that is, on instances whose protocol value
lies strictly between the soundness and completeness thresholds. {In this paper we focus only on $\QCIP[2]$, which is a two-message variant of $\QCIP$ and is defined as follows:}

\begin{definition}\label{def:qcip}
A promise problem $(L_{\mathrm{yes}},L_{\mathrm{no}})$ is in $\QCIP[2,c,s]$ if
there exist $\QPT$ algorithms $V_1,V_2$ such that:
\begin{itemize}
  \item \emph{Completeness:} for every $x\in L_{\mathrm{yes}}$, there is a
        family $\{z_y: y\in\{0,1\}^{\mathrm{poly}(|x|)}\}$ of
        $\mathrm{poly}(|x|)$-length strings with
        \begin{equation}
          \Pr_{(\rho,y)\leftarrow V_1(x)}[V_2(x,\rho,y,z_y)=1]\ \ge\ c(|x|).
        \end{equation}
  \item \emph{Soundness:} for every $x\in L_{\mathrm{no}}$ and every such
        family $\{z_y\}$,
        \begin{equation}
          \Pr_{(\rho,y)\leftarrow V_1(x)}[V_2(x,\rho,y,z_y)=1]\ \le\ s(|x|).
        \end{equation}
\end{itemize}
Henceforth, we let $\QCIP[2]:=\QCIP[2,2/3,1/3]$ to simplify notation.
\end{definition}

In the definition above, $\rho$ can be a quantum state, but $x,y,z_y$ are classical strings. Note that error reduction for $\QCIP[2]$ follows from parallel repetition: run
independent copies of the protocol in parallel and accept by majority vote.

For $\QPT$ algorithms $V_1,V_2$ and a string $x$, we write
\begin{equation}\label{eq:val}
  \val_{V_1,V_2}(x)\ :=\ \max_{\{z_y\}_y}\ \Pr_{(\rho,y)\leftarrow V_1(x)}\big[V_2(x,\rho,y,z_y)=1\big]
\end{equation}
for the value of the protocol on input $x$.

\begin{definition}\label{def:qcipOracle}
The set $\mathcal{O}_{\QCIP[2]}$ consists of any (computationally unbounded)
algorithm $O$ that satisfies the following condition. 
For any
$\QPT$ algorithms, 
$V_1$ and $V_2$, and any
$x\in\{0,1\}^*$,
\begin{equation}
  O(V_1,V_2,x)=
  \begin{cases}
    1 & \text{if } \val_{V_1,V_2}(x)\ \ge\ 2/3,\\[2pt]
    0 & \text{if } \val_{V_1,V_2}(x)\ \le\ 1/3.
  \end{cases}
\end{equation}
On inputs with $1/3<\val_{V_1,V_2}(x)<2/3$ the output of $O$ is unconstrained.
\end{definition}

Some care is needed when giving an algorithm {with} access to a $\QCIP[2]$ oracle,
because $\QCIP[2]$ is a class of promise problems. We do not fix a single
oracle but quantify over the whole set $\mathcal{O}_{\QCIP[2]}$. Concretely,
throughout this work a statement of the form ``$A^{O}$ cannot perform a task''
(for instance, distinguishing two distributions) means: there exists at least
one $O\in\mathcal{O}_{\QCIP[2]}$ for which $A^{O}$ fails. Symmetrically, $A^O$
counts as \emph{succeeding} only if it succeeds against \emph{every}
$O\in\mathcal{O}_{\QCIP[2]}$. Condition~4 of \cref{def:gap} is stated in exactly
this form.

The reason for this convention is that on an ambiguous input $(V_1,V_2,x)$ whose protocol value lies strictly between the two thresholds, i.e. $1/3  < \val_{V_1,V_2}(x) <2/3$,
it is not clear how the oracle should answer. This issue is not harmless: an oracle could smuggle useful
information into those answers---information that has nothing to do with
deciding the promise problem, yet that an algorithm might exploit. Quantifying
over all valid $O$ rules this out. An algorithm that succeeds only under a
particular choice of answers on ambiguous instances is reading off a feature of
one implementation rather than using the power to decide promise instances, and
does not count; only an algorithm that succeeds against every valid oracle is
using that deciding power. This is why $A^O$
counts as \emph{succeeding} only if it succeeds against \emph{every}
$O\in\mathcal{O}_{\QCIP[2]}$.

\subsection{Falsifiable Assumptions}

We recall the notion of falsifiable assumptions~\cite{Naor03,GW11,ADSS24,BMNY26}.

\begin{definition}[Falsifiable Cryptographic Assumptions]\label{def:falsifiable}
A falsifiable cryptographic assumption is a pair $(C,c)$ of a $\QPT$ algorithm
$C$ (challenger) and a constant $c\in[0,1)$. $C$ interacts with an adversary
over a quantum channel and then outputs $1/0$. The assumption is said to be
true if, for any $\QPT$ algorithm $A$,
\begin{equation}\label{eq:falsifiable}
  \Pr\big[1\leftarrow\langle C(1^n),A(1^n)\rangle\big]\ \le\ c+\negl(n),
\end{equation}
where $1\leftarrow\langle C(1^n),A(1^n)\rangle$ means that $C$ outputs $1$
after the interaction with $A$.
\end{definition}

\subsection{Witness Encryption}
\label{sec:we}

We recall the notion of witness encryption (\textsf{WE}) for $\QMA$ from~\cite{BM21}.
Initially, \textsf{WE} for \textsf{NP} was introduced in~\cite{GGSW13}. We define encryption only
with respect to classical messages. Also, we only define it for single-bit messages,
but it is easy to generalize the definition to multiple bits.

\begin{definition}[Witness Encryption for $\QMA$]\label{def:we}
Let $L\in\QMA$ and $Q$ be its verification algorithm. A \textsf{WE}
scheme $(\WEEnc,\WEDec)$ for the relation $R_{L,Q}$ consists of the following
$\QPT$ algorithms:
\begin{itemize}
  \item $\WEEnc(1^n,x,m)$: On input a statement $x\in \{0,1\}^n$ and a message
        $m\in\{0,1\}$, it outputs a classical ciphertext $\ct$.
  \item $\WEDec(x,\ct,|\psi\rangle)$: On input a statement $x$, a classical
        ciphertext $\ct$, and a quantum state $|\psi\rangle$, it outputs a
        message $\widetilde{m}\in\{0,1\}$ or $\bot$.
\end{itemize}
\end{definition}

We define correctness below. Throughout, $t$ denotes the number of witness
copies consumed by decryption.

\begin{definition}[Correctness]\label{def:weCorrect}
A \textsf{WE} $(\WEEnc,\WEDec)$ for a relation $R_{L,Q}$ is correct if
there exists a negligible function $\nu(n)$ and a polynomial $t(n)$ such that
for all $m\in\{0,1\}$, all sequences of instances $\{x_n\}_{n\in\mathbb{N}}$,
where $|x_n|=n$ for each $n\in\mathbb{N}$, and witnesses
$\{|\psi_n\rangle\}_{n\in\mathbb{N}}$, where each $(x_n,|\psi_n\rangle)\in R_{L,Q}$,
it holds that
\[
  \Pr\!\left[\WEDec\big(x_n,\WEEnc(1^n,x_n,m),|\psi_n\rangle^{\otimes t(n)}\big)=m\right]\ \ge\ 1-\nu(n).
\]
\end{definition}

{Next} we present the definition of security of \textsf{WE} against
quantum algorithms. Our definition is a bit weaker than the traditional
definition, but this just makes our separation stronger. Specifically, the
traditional notion requires that for any $x\notin L$, encryptions of $0$ and
$1$ are computationally indistinguishable. Instead, we allow the adversary to
choose $x$, and it wins the security game if it manages to distinguish the
encrypted message and $x\notin L$. 

\begin{definition}[Security]\label{def:weSec}
A \textsf{WE} scheme $\Pi:=(\WEEnc,\WEDec)$ for a relation $R_{L,Q}$ is secure
if for any $\QPT$ adversary $A:=(A_1,A_2)$,
\begin{equation}
  \Pr[\WESec_{\Pi,A}(1^n)=1]\ \le\ \tfrac{1}{2}+\negl(n).
\end{equation}
\end{definition}

\begin{center}
\fbox{\begin{minipage}{0.92\textwidth}
$\WESec_{\Pi,A}(1^n)$:
\begin{enumerate}
  \item $b\leftarrow\{0,1\}$.
  \item $x\leftarrow A_1(1^n)$. %{Here, $\sigma$ is a quantum state.}
  \item $\ct\leftarrow\WEEnc(1^n,x,b)$.
  \item $b'\leftarrow A_2(x,\ct)$.
  \item The output of the experiment is $1$ if $x\in \overline{L}_n$ and $b'=b$.
        Otherwise, the output is $0$.
\end{enumerate}
\end{minipage}}
\end{center}

\begin{theorem}[Lemma 3.15 in~\cite{BM21}]\label{thm:bm}
Assuming null--\textsf{qIO}, there exists \textsf{WE} for $\QMA$.
\end{theorem}

\subsection{Quantum Black-Box Reductions}\label{sec:bbred}

Quantum black-box reductions are explored in several
works~\cite{LP24,TZ25,CM24,CCS24}. We formulate quantum black-box reductions
specifically showing the security of a \textsf{WE} scheme based on a falsifiable
cryptographic assumption.

\begin{definition}[$\Pi$-adversary]
\label{def:adversary}
Let $L\in\QMA$ and $Q$ be its verification algorithm, and let
$\Pi=(\WEEnc,\WEDec)$ be a \textsf{WE} scheme for the relation
$R_{L,Q}$. We say that a quantum algorithm $\overline{P}=(\overline{P}_1,\overline{P}_2)$ is a
$\Pi$-adversary if there exists a polynomial $p$ such that
\begin{equation}
  \Pr\big[\WESec_{\Pi,\overline{P}}(1^n)=1\big]\ \ge\ \tfrac{1}{2}+\tfrac{1}{p(n)}
\end{equation}
for infinitely many $n\in\mathbb{N}$. In this case, we say $\overline{P}$ has
advantage $1/p$.
\end{definition}

%A $\Pi$-adversary $\overline{P}=(\overline{P}_1,\overline{P}_2)$ is accessed as a \emph{stateful} oracle: a query to $\overline{P}_1$ returns only the instance $x$, while the state $\sigma$ produced by $\overline{P}_1$ is kept by the oracle and used to answer the matching query to $\overline{P}_2$. The reduction never receives $\sigma$.

\begin{definition}[\textsf{WE}--reduction]\label{def:reduction}
Let $\Pi$ be a \textsf{WE} scheme for $\QMA$. A quantum black-box {reduction}
showing the security of $\Pi$ from a falsifiable cryptographic assumption
$(C,c)$ is a $\QPT$ algorithm $\Sigma^{(\cdot)}$ such that for any polynomial
$p$, there exists a polynomial $p'$ such that for any (even inefficient)
$\Pi$-adversary $\overline{P}$ with advantage $1/p$, the algorithm
$\Sigma^{\overline{P}}$ breaks the assumption with advantage $1/p'$, i.e.
\begin{equation}\label{eq:redbreaks}
  \Pr\big[1\leftarrow\langle C(1^n),\Sigma^{\overline{P}}(1^n)\rangle\big]\ \ge\ c+\frac{1}{p'(n)},
\end{equation}
for infinitely many $n\in\mathbb{N}$. 
For short, we call such a reduction a \emph{\textsf{WE}--reduction}.
\end{definition}

In this work, we consider only reductions $\Sigma$ that query the
adversary on the same security parameter, non-adaptively, classically, and on honestly
generated ciphertexts. More formally, our separation applies to reductions defined as follows, 
which we call \emph{{restricted} \textsf{WE}--reduction}.

\begin{definition}[Restricted \textsf{WE}--reduction]\label{def:wered}
Let $(C,c)$ be a falsifiable assumption, and let $\Sigma^{(\cdot)}$ be a
\textsf{WE}--reduction defined in \cref{def:reduction}. We say that
$\Sigma^{(\cdot)}$ is a \emph{restricted \textsf{WE}--reduction} if there
exist polynomials $s=s(n)$ and $\ell=\ell(n)$ such that, in the
interaction
$
  \big\langle C(1^n),\Sigma^{\overline P}(1^n)\big\rangle,
$
the oracle access of $\Sigma$ to
$\overline P=(\overline P_1,\overline P_2)$ has the following form.
\begin{enumerate}
\item The challenger $C$ and the reduction $\Sigma$ may interact arbitrarily
throughout the execution. Only $\Sigma$ has oracle access to $\overline P$.

\item For each $i\in[s]$, $\Sigma$ submits $1^n$ to $\overline P_1$ and receives
an instance $x_i$. 

\item For each $i \in [s]$ and $j \in [\ell]$, a bit $m_{i,j} \gets
        \{0,1\}$ is sampled uniformly at random and the ciphertext $\mathsf{ct}_{i,j}
        \leftarrow \mathsf{WE.Enc}(1^n, x_i, m_{i,j})$ is computed by running the encryption algorithm honestly. This step may be performed by $\Sigma$, by $C$, or jointly through their interaction. 

\item For each $i \in [s]$ and $j \in [\ell]$,  $\Sigma$ submits {each $x_i$ and} $\ct_{i,j}$ to the oracle $\overline P_2$ and receives a
response $\widetilde m_{i,j}\gets \overline{P}_2(x_i,\ct_{i,j})$. 

\item The execution ends when $C$ outputs its decision
bit.
\end{enumerate}

Thus, the oracle queries are {parameter-preserving}, classical, and {non-adaptive}. 
\end{definition}

\begin{remark}[Honestly generated ciphertexts]
Note that in \cref{def:wered}, we assume that ciphertext queries are honestly generated, meaning they are generated by running the encryption algorithm honestly on a randomly sampled plaintext bit.
We believe that this restriction is justified because, intuitively, the
reduction should turn a valid attack against the \textsf{WE} scheme into
a valid attack against the falsifiable assumption. A valid attack against
\textsf{WE} involves distinguishing or breaking the security of an
honestly generated ciphertext. Therefore, a reduction should be able to use this property alone to carry out a successful attack on the underlying falsifiable assumption, without needing to query $\overline{P}$ on maliciously
generated ciphertexts. 
\end{remark}

%\begin{remark}[Why the state is kept by the oracle]\label{rem:stateful}
%The convention that $\Sigma$ never receives the register $I$\mor{what is $I$? It did not appear in the above definition} is needed, and it is also the natural one: $I$ contains the adversary's private state, and a black-box reduction interacts with the adversary only through its input/output behavior. If instead $\Sigma$ were handed $I$, then no simulator could work. Indeed, the $\QPT$ simulator $P$ we build simply decrypts the ciphertexts and its internal state contains $\QMA$ witnesses. If $\Sigma$ receives these witnesses, it can run the verification algorithm and distinguish $P$ and $\overline{P}$. \end{remark}

Since a restricted \textsf{WE}--reduction only queries non-adaptively, we can split the execution in \cref{def:wered} into multiple phases that will be useful in our analysis later on.

\begin{definition}[Decomposition of $\Gamma$]\label{def:decomposition}
Let $(C,c)$ be a falsifiable assumption, and let $\Sigma^{(\cdot)}$ be a
restricted \textsf{WE}--reduction as in \cref{def:wered}. 
For any adversary $\overline{P}=(\overline{P}_1,\overline{P}_2)$, 
let
\begin{align}
\Gamma^{\overline{P}}(1^n) := \langle C(1^n), \Sigma^{\overline{P}}(1^n)\rangle
\end{align}
denote
the algorithm obtained by viewing $C$ and $\Sigma$ together as a single
algorithm; its output is the decision bit of $C$. 
We define the following two algorithms, $\Gamma_1,\Gamma_2$, and consider the following \emph{phased execution}:
\begin{itemize}
\item \textbf{Phase 1:} Run $\overline{P}_1(1^n)$ independently $s$ times to
generate $x=(x_1,\dots,x_s)$. The algorithm $\overline{P}_1$ is not executed in
any later phase.
\item \textbf{Phase 2:} $\Gamma_1(1^n,x)$ simulates the execution of
$\Gamma^{\overline{P}}(1^n)$, including the interaction between $C$ and $\Sigma$,
until just before $\Sigma$ makes its first query to $\overline{P}_2$. When
$\Sigma$ makes its $i$-th query to $\overline{P}_1$, $\Gamma_1$ answers it with
the pre-generated instance $x_i$ {in Phase 1} instead of executing $\overline{P}_1$. The
output of this phase consists of the pairs
$(\ct_{i,j})_{i\in[s],j\in[\ell]}$ generated in this
execution, together with the joint internal state $st$ of $C$ and
$\Sigma$ at this point.
\item \textbf{Phase 3:} For each $i\in[s]$ and $j\in[\ell]$, run $\overline{P}_2$
on $(x_i,\ct_{i,j})$, and let $\widetilde{m}_{i,j}$ be the
response.
\item \textbf{Phase 4:} $\Gamma_2(st,x,\mathsf{ct},\widetilde{m})$
resumes the execution of $\Gamma^{\overline{P}}(1^n)$ from the state $st$. When
$\Sigma$ queries $\overline{P}_2$ on $\mathsf{ct}_{i,j}$, $\Gamma_2$ answers it with the recorded response $\widetilde{m}_{i,j}$ 
{obtained in Step 3} instead of executing
$\overline{P}_2$. The phase ends when $C$ outputs its decision bit, which is the
output of the phased execution.
\end{itemize}
\end{definition}

\begin{remark}\label{rem:state}
The state $st$ is determined by the execution of Phase~2. When no
confusion arises, we suppress it from the notation and write
$\Gamma_2(x,\mathsf{ct},\widetilde{m})$.
\end{remark}

\begin{lemma}\label{lem:decomposition}
Let $\Sigma^{(\cdot)}$ be a restricted \textsf{WE}--reduction. Then the
output of the phased execution of \cref{def:decomposition} is
distributed identically to the output of $\Gamma^{\overline{P}}(1^n)$.
\end{lemma}
\begin{proof}
The phased execution differs from $\Gamma^{\overline{P}}(1^n)$ only in the timing
of the executions of $\overline{P}_1$ and $\overline{P}_2$. Each execution of $\overline{P}_1$ takes the
input $1^n$ and fresh randomness only, so its output is independent of
the preceding transcript; the $s$ executions may therefore be moved to
the beginning of the execution, and answering the queries of $\Sigma$
with the pre-generated $x_1,\dots,x_s$ leaves the joint distribution of
the execution unchanged. Similarly, each honest generation of
$\mathsf{ct}_{i,j}$ is sampled from a distribution that depends only on $x_i$, so it may be
moved before the first query to $\overline{P}_2$. Finally, since the queries to
$\overline{P}_2$ are non-adaptive, the list $(x_i,\mathsf{ct}_{i,j})_{i,j}$ is
determined at the end of Phase~2, and by assumption each response of
$\overline{P}_2$ depends only on its query and on randomness independent of the
transcript; hence computing all responses in Phase~3 and replaying them
in Phase~4 also leaves the joint distribution unchanged. In particular,
the decision bit of $C$ has the same distribution in both executions.
\qed \end{proof}

\subsection{$\Gap$ Gap Problem}

We introduce the notion of $\Gap$ gap problems. We will assume the
existence of these problems in our separation.

\begin{definition}[$\Gap$ gap problems]\label{def:gap}
A $\Gap$ gap problem of a
language $L$ in $\QMA$ with verification algorithm $Q$ and a $\QMA$ relation
$R_{L,Q}$ is a pair $(\bL,\bLb)$ satisfying the following conditions:
\begin{enumerate}
  \item $\bL:=\{\bL_n\}_{n\in\mathbb{N}}$ and
        $\bLb:=\{\bLb_n\}_{n\in\mathbb{N}}$ are distribution ensembles, where
        $\bL_n$ is a distribution over $L_n:=L\cap\{0,1\}^n$ and $\bLb_n$ is a
        distribution over $\overline{L}_n:=\{0,1\}^n\setminus L_n$ for each
        $n\in\mathbb{N}$.
  \item There exists a $\QPT$ algorithm $\SampYes$ that takes as input $(1^n,1^t)$, for an arbitrary polynomial $t=t(n)$,
        and outputs $(x,|w_x\rangle^{\otimes t})$, where
        $(x,|w_x\rangle)\in R_{L,Q}$, such that the projection to the first
        coordinate of the output is $\bL_n$.
  \item There exists a not-necessarily-efficient algorithm $\SampNo$ that takes
        as input $1^n$ and samples from the distribution $\bLb_n$.
\item For every polynomial $s$ and every $\QPT$ distinguisher
        $A$, there exists an oracle $O\in\mathcal{O}_{\QCIP[2]}$ such that
        \begin{equation}\label{eq:gapind}
          \Big|\Pr_{(x_1,...,x_s)\leftarrow\bL_n^{\otimes s}}[A^{O}(x_1,...,x_s)=1]
          -\Pr_{(x_1,...,x_s)\leftarrow\bLb_n^{\otimes s}}[A^{O}(x_1,...,x_s)=1]\Big|\ \le\ \negl(n),
        \end{equation}
        where $\bL_n^{\otimes s}$ (resp.\ $\bLb_n^{\otimes s}$) denotes $s$
        independent samples of $\bL_n$ (resp.\ $\bLb_n$).
\end{enumerate}
When $s$ independent instances are needed, we write $\SampYes(1^n,1^t,1^s)$
for the algorithm that runs $s$ independent copies of $\SampYes(1^n,1^t)$, and
similarly for $\SampNo$.
\end{definition}

Note that we only allow for classical access to the $\QCIP[2]$ oracle in this
work. This can only strengthen our separation, as it weakens the assumption of
a $\Gap$ gap problem.

\section{Simulating the Ideal Decryptor}\label{sec:main}

In this section, we prove \cref{thm:main} below, which is the main simulation result
used in our separation.

For any challenger $C$ of a falsifiable assumption and a restricted
\textsf{WE}--reduction $\Sigma$, let $\Gamma$ be the interaction between the two algorithms and let $\Gamma_1,\Gamma_2$ be the decomposition as defined in \cref{def:decomposition}.

Let $L$ be a $\QMA$
language with a $\Gap$ gap problem $(\bL,\bLb)$ and verification algorithm
$Q$. Let $\Pi=(\WEEnc,\WEDec)$ be a \textsf{WE} scheme for $R_{L,Q}$ that
uses at most $t=t(n)$ witness copies for decryption.

We now define the two experiments that will be compared. Both experiments depend on polynomials $s,\ell$ which are determined by $\Sigma$ in the same way as \cref{def:wered}. They also depend on $t$ which is determined by $\Pi$.  

\begin{figure}[ht]
\begin{center}
\fbox{\begin{minipage}{0.92\textwidth}
$\DecYes^{\SampYes}_{\Gamma,D}(1^n)$:
\begin{enumerate}
  \item Sample
  \[
      (x_i,|\psi_i\rangle^{\otimes t\ell})_{i\in[s]}
      \leftarrow\SampYes(1^n,1^{t\ell},1^s),
  \]
  and set $x=(x_1,\ldots,x_s)$.

  \item Run
  \[
      \big(\ct_{i,j}\big)_{i,j}
      \leftarrow\Gamma_1(1^n,x).
  \]

  \item For every $i\in[s]$ and $j\in[\ell]$, compute
  \[
      a_{i,j}\leftarrow
      D(x_i,\ct_{i,j},|\psi_i\rangle^{\otimes t}),
  \]
  using fresh witness copies for each query.

  \item Output $\Gamma_2(x,\ct,a)$.
\end{enumerate}
\end{minipage}}
\end{center}
\caption{The yes-based experiment.}\label{fig:decyes}
\end{figure}

\begin{figure}[ht]
\begin{center}
\fbox{\begin{minipage}{0.92\textwidth}
$\DecNo^{\SampNo}_{\Gamma,S}(1^n)$:
\begin{enumerate}
  \item Sample
  \[
      (x_1,\ldots,x_s)\leftarrow\SampNo(1^n,1^s),
  \]
  and set $x=(x_1,\ldots,x_s)$.

  \item Run
  \[
      \big(\ct_{i,j}\big)_{i,j}
      \leftarrow\Gamma_1(1^n,x).
  \]

  \item For every $i\in[s]$ and $j\in[\ell]$, compute
  \[
      a_{i,j}\leftarrow S(x_i,\ct_{i,j}).
  \]

  \item Output $\Gamma_2(x,\ct,a)$.
\end{enumerate}
\end{minipage}}
\end{center}
\caption{The no-based experiment.}\label{fig:decno}
\end{figure}

\begin{theorem}\label{thm:main}
There exist a computationally unbounded classical function
$S^*$ and a QPT algorithm $D^*$, both depending only on $\Pi$,
such that the following holds. 

Let $\epsilon$ be an inverse polynomial. For every falsifiable assumption $(C,c)$, and every
restricted \textsf{WE}--reduction $\Sigma$ showing the security of $\Pi$ from
$(C,c)$, let $\Gamma$ be the interaction $\langle C,\Sigma\rangle$ and let $(\Gamma^1,\Gamma^2)$ be its decomposition
(see Definition~\ref{def:decomposition}). Then, for all sufficiently large $n$,
\begin{equation}\label{eq:main}
  \Big|
  \Pr\big[\DecYes^{\SampYes}_{\Gamma,D^*}(1^n)=1\big]
  -
  \Pr\big[\DecNo^{\SampNo}_{\Gamma,S^*}(1^n)=1\big]
  \Big|
  \le \epsilon(n)
\end{equation}
for all sufficiently large $n\in\mathbb{N}$.
In fact, $D^*$ can be taken as the decryption algorithm, $\WEDec$. 
\end{theorem}

Before proving \cref{thm:main}, we establish the ingredients needed for the
simulation.

\subsection{The Bayes Decryptor}

We now define $S^*$. It is determined entirely by $\WEEnc$ and
does not depend on the reduction, the challenger, or $\epsilon$.

\begin{definition}[Bayes decryptor and Bayes error]\label{def:bayes}
Let $x\in\{0,1\}^n$. Let $\textbf{E}_x$ be the distribution of the pair
$(m,\ct)$ obtained by sampling a uniform bit $m\leftarrow\{0,1\}$ and then
sampling $\ct\leftarrow\WEEnc(1^n,x,m)$.

For $b\in\{0,1\}$, define
\begin{equation}\label{eq:posterior}
p_b(x,\ct):=
\frac{\Pr[\WEEnc(1^n,x,b)=\ct]}
{\Pr[\WEEnc(1^n,x,0)=\ct]+\Pr[\WEEnc(1^n,x,1)=\ct]}.
\end{equation}
If the denominator is zero, set
$p_0(x,\ct)=p_1(x,\ct):=1/2$. Thus $p_b(x,\ct)$ is the posterior probability
that the plaintext was $b$, conditioned on observing $\ct$. In particular,
\[
    p_0(x,\ct)+p_1(x,\ct)=1
\]
for every $\ct$.

The \emph{Bayes decryptor} is the deterministic classical function
\begin{equation}\label{eq:bayes-decryptor}
S^*(x,\ct):=\argmax_{b\in\{0,1\}}p_b(x,\ct),
\end{equation}
where ties are broken towards $0$. Its \emph{Bayes error} is
\begin{equation}\label{eq:bayes-error}
\Delta(x):=
\Pr_{(m,\ct)\leftarrow\textbf{E}_x}
\big[S^*(x,\ct)\neq m\big].
\end{equation}
\end{definition}

Conditioned on the ciphertext, $S^*$ errs with probability
$\min_{b\in\{0,1\}}p_b(x,\ct)$. Therefore,
\begin{equation}\label{eq:bayes-error-alt}
\Delta(x)
=
\Ex_{\ct}
\left[\min_{b\in\{0,1\}}p_b(x,\ct)\right]
\le \frac12,
\end{equation}
where $\ct$ is distributed as the second coordinate of $\textbf{E}_x$.
Evaluating $S^*$ may require unbounded computation, which is allowed in
\cref{thm:main}.

The next two elementary lemmas will be used repeatedly.

\begin{lemma}[Optimality of $S^*$]\label{lem:opt}
Let $x\in\{0,1\}^n$, let $\tau_x$ be any quantum state depending only on $x$,
and let $A$ be any, possibly computationally unbounded, quantum algorithm.
Then
\begin{equation}
  \Pr_{(m,\ct)\leftarrow\textbf{E}_x}
  \big[A(x,\ct,\tau_x)=m\big]
  \le 1-\Delta(x).
\end{equation}
\end{lemma}

\begin{proof}
Condition on $x$ and $\ct$. Since $\tau_x$ depends only on $x$, the
plaintext $m$ is independent of $\tau_x$ given $(x,\ct)$. Let
$\alpha_b:=\Pr[b\leftarrow A(x,\ct,\tau_x)]$ for $b\in\{0,1\}$. The
conditional probability of recovering $m$ is therefore
\[
    \sum_{b\in\{0,1\}}\alpha_b p_b(x,\ct)
    \le \max_{b\in\{0,1\}}p_b(x,\ct).
\]
Averaging over $\ct$ and using \cref{eq:bayes-error-alt} proves the claim.
\qed \end{proof}

\begin{lemma}[Per-message error]\label{lem:permsg}
For every $x\in\{0,1\}^n$ and every $m\in\{0,1\}$,
\begin{equation}
  \Pr_{\ct\leftarrow\WEEnc(1^n,x,m)}
  \big[S^*(x,\ct)\neq m\big]
  \le 2\Delta(x).
\end{equation}
\end{lemma}

\begin{proof}
By \cref{eq:bayes-error},
\[
\Delta(x)
=
\frac12
\sum_{m\in\{0,1\}}
\Pr_{\ct\leftarrow\WEEnc(1^n,x,m)}
\big[S^*(x,\ct)\neq m\big].
\]
Both summands are non-negative, so each is at most $2\Delta(x)$.
\qed \end{proof}

\subsection{Decryptability on No-Instances}

Our first main step is to show that the Bayes error is small for almost every
no-instance drawn from $\bLb_n$. Thus, although these instances have no
witnesses, their honest ciphertexts can still be decrypted with small error by
the possibly inefficient function $S^*$.

\begin{lemma}[Statistical decryptability on no-instances]\label{lem:decoff}
For every inverse polynomial $\gamma$,
\begin{equation}
  \delta_{n,\gamma}
  :=
  \Pr_{x\leftarrow\bLb_n}
  \big[\Delta(x)\ge\gamma(n)\big]
  \le \negl(n).
\end{equation}
\end{lemma}

\begin{proof}
We first show that the Bayes error is negligible on every yes-instance.

\smallskip
\noindent\textbf{Step 1: $\Delta(x)\le\nu(n)$ for every
$x\in L\cap\{0,1\}^n$.}
Let $x\in L\cap\{0,1\}^n$, and let $|w_x\rangle$ satisfy
$(x,|w_x\rangle)\in R_{L,Q}$. Apply \cref{lem:opt} with
$\tau_x=|w_x\rangle^{\otimes t}$ and $A=\WEDec$. By correctness,
$\WEDec$ recovers a uniformly random plaintext with probability at least
$1-\nu(n)$. Hence
\[
    1-\nu(n)\le 1-\Delta(x),
\]
and therefore $\Delta(x)\le\nu(n)$.

\smallskip
\noindent\textbf{Step 2: a two-message protocol.}
Fix $\gamma$ and set
\[
    r:=\left\lceil\frac{16}{\gamma(n)}\right\rceil.
\]
Since $\gamma$ is inverse polynomial, $r$ is polynomial. Define $\QPT$
algorithms $V_1,V_2$ as follows:
\begin{itemize}
  \item $V_1(x)$: for each $j\in[r]$, sample
  $m_j\leftarrow\{0,1\}$ and
  $\ct_j\leftarrow\WEEnc(1^n,x,m_j)$. Send
  \[
      y=(\ct_1,\ldots,\ct_r)
  \]
  and keep
  \[
      \rho=(m_1,\ldots,m_r)
  \]
  private.

  \item $V_2(x,\rho,y,z_y)$: parse
  $z_y=(b_1,\ldots,b_r)\in\{0,1\}^r$ and accept if and only if
  \[
      \big|\{j\in[r]:b_j\neq m_j\}\big|
      \le \frac{\gamma r}{2}.
  \]
\end{itemize}

\smallskip
\noindent\textbf{Step 3: the value of the protocol.}
We claim that, for every $x$,
\begin{equation}\label{eq:protoval}
  \val_{V_1,V_2}(x)
  =
  \Pr\left[Z\le\frac{\gamma r}{2}\right],
  \qquad
  Z\sim\Bin(r,\Delta(x)).
\end{equation}

Fix a verifier message $y=(\ct_1,\ldots,\ct_r)$ and write
\[
    \mu_j:=\min_{b\in\{0,1\}}p_b(x,\ct_j).
\]
The pairs $(m_j,\ct_j)$ are independent samples from $\textbf{E}_x$.
Therefore, conditioned on $y$, the plaintext bits $m_1,\ldots,m_r$ remain
independent and
\[
    \Pr[m_j=b\mid y]=p_b(x,\ct_j).
\]

Now fix a response $z=(b_1,\ldots,b_r)$. Conditioned on $y$, the error in
position $j$ has probability
\[
    q_j:=1-p_{b_j}(x,\ct_j)\ge\mu_j,
\]
with equality when $b_j=S^*(x,\ct_j)$. Let
$U_1,\ldots,U_r$ be independent uniform random variables on $[0,1]$ and set
\[
    X_j=\mathbf{1}[U_j\le q_j],
    \qquad
    X_j^*=\mathbf{1}[U_j\le\mu_j].
\]
Then $X_j^*\le X_j$ pointwise for every $j$. Hence
\[
\Pr\left[\sum_jX_j\le\frac{\gamma r}{2}\,\middle|\,y\right]
\le
\Pr\left[\sum_jX_j^*\le\frac{\gamma r}{2}\,\middle|\,y\right].
\]
Thus, for every fixed $y$, the optimal response is
\[
    z_y^*
    =
    \big(S^*(x,\ct_1),\ldots,S^*(x,\ct_r)\big).
\]
By the definition of $\QCIP[2]$, the prover may choose its response separately
for each verifier message $y$, so this pointwise choice attains the protocol
value.

Under this strategy, the error indicators are
$\mathbf{1}[S^*(x,\ct_j)\neq m_j]$. Before conditioning on $y$, these are
independent Bernoulli variables with mean $\Delta(x)$. Their sum is therefore
distributed as $\Bin(r,\Delta(x))$, proving \cref{eq:protoval}.

\smallskip
\noindent\textbf{Step 4: the two cases.}
For all sufficiently large $n$:
\begin{itemize}
  \item If $x\in L\cap\{0,1\}^n$, then $\Delta(x)\le\nu(n)$ by Step~1.
  Hence $\Ex[Z]\le\nu r$, and Markov's inequality gives
  \[
      \Pr\left[Z>\frac{\gamma r}{2}\right]
      \le \frac{2\nu}{\gamma}
      =\negl(n).
  \]
  Therefore $\val_{V_1,V_2}(x)\ge2/3$.

  \item If $x\notin L$ and $\Delta(x)\ge\gamma$, then
  $\Ex[Z]=\Delta(x)r\ge\gamma r$. Thus
  $\gamma r/2\le\Ex[Z]/2$, and the multiplicative Chernoff bound gives
  \[
      \Pr\left[Z\le\frac{\gamma r}{2}\right]
      \le
      e^{-\Ex[Z]/8}
      \le
      e^{-\gamma r/8}
      \le e^{-2}
      <\frac13.
  \]
  Hence $\val_{V_1,V_2}(x)\le1/3$.
\end{itemize}

\smallskip
\noindent\textbf{Step 5: applying the gap assumption.}
Let $A^O(x)$ submit $(V_1,V_2,x)$ to the oracle $O$ and output the returned
bit. This is a $\QPT$ algorithm making one classical oracle query. By
\cref{def:qcipOracle} and Step~4, for all sufficiently large $n$ and every
$O\in\mathcal{O}_{\QCIP[2]}$,
\[
    \Pr_{x\leftarrow\bL_n}[A^O(x)=1]=1
\]
and
\[
    \Pr_{x\leftarrow\bLb_n}[A^O(x)=1]
    \le 1-\delta_{n,\gamma}.
\]
Condition~4 of \cref{def:gap} gives an oracle
$O\in\mathcal{O}_{\QCIP[2]}$ for which these two probabilities differ by at
most a negligible amount. For that oracle their difference is at least
$\delta_{n,\gamma}$. Therefore
\[
    \delta_{n,\gamma}\le\negl(n),
\]
as required.
\qed \end{proof}

\subsection{Estimating the Ideal Experiment}\label{sec:ideal}

Fix the two-stage process $\Gamma=(\Gamma_1,\Gamma_2)$ defined at the
beginning of this section. Recall that every real plaintext is sampled as a
fresh uniform bit. More precisely, for every
$i\in[s]$ and $j\in[\ell]$, the corresponding real ciphertext is generated as
\[
    \widetilde m_{i,j}\leftarrow\{0,1\},
    \qquad
    \ct_{i,j}\leftarrow
    \WEEnc(1^n,x_i,\widetilde m_{i,j}).
\]

For an instance tuple $x=(x_1,\ldots,x_s)$, define the \emph{ideal
experiment} $\Ideal_\Gamma(x)$ as follows. Run
\[
    (\ct_{i,j})_{i,j}
    \leftarrow \Gamma_1(1^n,x),
\]
and let $\widetilde m=(\widetilde m_{i,j})_{i,j}$ be the uniform plaintext
bits used to generate the real ciphertexts in this execution. Instead of
querying a decryptor, give $\widetilde m$ to the second stage in its original
$(i,j)$ positions and output
\[
    \Gamma_2(x,\ct,\widetilde m).
\]
Write
\begin{equation}\label{eq:iota}
    \iota(x):=\Pr[\Ideal_\Gamma(x)=1].
\end{equation}

The ideal experiment is only a proof device. In particular, the plaintext
vector need not be available to the reduction or to an efficient algorithm.
We next show that $\iota(x)$ can nevertheless be estimated by a $\QPT$
algorithm with access to a $\QCIP[2]$ oracle.

To do this, the verifier hides the real ciphertexts among test ciphertexts. It does not
need to know the real plaintexts. In each block, it adds the same number of
test encryptions of $0$ and of $1$.

Fix an even polynomial $K=K(n)$ with $K\geq\ell$. Define $\QPT$ algorithms
$(V^\Gamma_1,V^\Gamma_2)$ as follows.

\begin{center}
\fbox{\begin{minipage}{0.92\textwidth}
$V^\Gamma_1(x)$ on input $x=(x_1,\ldots,x_s)$:
\begin{enumerate}
  \item Run
  \[
      (\ct_{i,j})_{i,j}
      \leftarrow\Gamma_1(1^n,x).
  \]
  These are the \emph{real} ciphertexts. Their plaintexts are not needed by
  the verifier.

  \item For every $i\in[s]$, generate $K/2$ test encryptions of $0$ and
  $K/2$ test encryptions of $1$:
  \[
      \ct'_{i,c,r}\leftarrow\WEEnc(1^n,x_i,c),
      \qquad
      c\in\{0,1\},\quad r\in[K/2].
  \]
  Every test ciphertext is generated honestly.

  \item For every $i\in[s]$, choose an independent uniformly random
  permutation $\pi_i$ of the $K+\ell$ ciphertexts in block $i$. Let $y_i$
  be the resulting shuffled list.

  \item Send $y=(y_1,\ldots,y_s)$. Keep the permutations, the
  real and test positions, and the plaintexts of the test ciphertexts
  private in a state $\rho$.
\end{enumerate}

\medskip
$V^\Gamma_2(x,\rho,y,z_y)$: parse $z_y$ as one bit for every shuffled
ciphertext. Reject if any test ciphertext is answered incorrectly, i.e. does not match its plaintext. Otherwise,
extract the answers corresponding to the real ciphertexts, restore their
original $(i,j)$ order, call the resulting vector $a$, and output
$\Gamma_2(x,\ct,a)$.
\end{minipage}}
\end{center}

\begin{lemma}[Value of the hidden-test protocol]\label{lem:hidden}
For every $x=(x_1,\ldots,x_s)$,
\begin{equation}\label{eq:hidden-upper}
    \val_{V^\Gamma_1,V^\Gamma_2}(x)
    \leq
    \iota(x)+\frac{2s\ell}{K}.
\end{equation}
Moreover, if $\Delta(x_i)\leq\gamma$ for every $i\in[s]$, then
\begin{equation}\label{eq:hidden-lower}
    \val_{V^\Gamma_1,V^\Gamma_2}(x)
    \geq
    \iota(x)-s(K+\ell)\gamma
    \geq
    \iota(x)-2sK\gamma.
\end{equation}
\end{lemma}

\begin{proof}
Let $W$ be the event that every answer, on both real and test ciphertexts,
equals the corresponding true plaintext.

Couple the protocol and the ideal experiment so that they use the same
execution of $\Gamma_1$. In particular, they use the same real ciphertexts,
the same latent plaintext vector $\widetilde m$. On the event $W$, the vector supplied to $\Gamma_2$ is exactly
$\widetilde m$ in its original order. Therefore,
\begin{equation}\label{eq:Wevents}
    \Pr[V^\Gamma_2\text{ accepts}\wedge W]
    =
    \Pr[\Ideal_\Gamma(x)=1\wedge W].
\end{equation}

\smallskip
\noindent\emph{Lower bound.}
Consider the prover that answers every ciphertext $c$ in block $i$ with
$S^*(x_i,c)$. For $c\in\{0,1\}$, write
\[
    e_{i,c}
    :=
    \Pr_{\ct\leftarrow\WEEnc(1^n,x_i,c)}
    [S^*(x_i,\ct)\neq c].
\]
By the definition of the Bayes error,
\[
    \Delta(x_i)=\frac{e_{i,0}+e_{i,1}}{2}.
\]
At each real position in block $i$, conditioned on the preceding transcript,
the plaintext is uniform and the ciphertext is a fresh honest encryption.
Consequently, the error probability at that position is $\Delta(x_i)$, and
the sum of the error probabilities at the $\ell$ real positions is
$\ell\Delta(x_i)$. The block also contains $K/2$ tests of each bit,
whose error probabilities sum to $K\Delta(x_i)$. Thus, a union bound over all
ciphertexts gives
\[
    \Pr[\neg W]
    \leq
    \sum_{i=1}^s(K+\ell)\Delta(x_i)
    \leq
    s(K+\ell)\gamma
    \leq
    2sK\gamma,
\]
where the last inequality uses $K\geq\ell$. Using \cref{eq:Wevents},
\begin{align}
\val_{V^\Gamma_1,V^\Gamma_2}(x)
&\geq
\Pr[V^\Gamma_2\text{ accepts}\wedge W]\\
&=
\Pr[\Ideal_\Gamma(x)=1\wedge W]\\
&\geq
\iota(x)-s(K+\ell)\gamma.
\end{align}

\smallskip
\noindent\emph{Upper bound.}
Fix an arbitrary prover strategy. Since the verifier rejects whenever a test
ciphertext is answered incorrectly,
\begin{align}
\Pr[\mathrm{accept}]
&\leq
\Pr[\mathrm{accept}\wedge W]
+
\Pr[\text{all tests correct}\wedge\neg W]\\
&\leq
\iota(x)+\Pr[\mathcal B],
\label{eq:hidden-B}
\end{align}
where $\mathcal B$ is the event that every test ciphertext is answered
correctly but at least one real ciphertext is answered incorrectly.

Fix a block $i$. For $c\in\{0,1\}$, let
\[
    R_{i,c}
    :=
    \big|\{j\in[\ell]:\widetilde m_{i,j}=c\}\big|.
\]
This quantity is used only in the analysis. The verifier does not need to
know it.

For the purpose of an upper bound, reveal to the prover the plaintext label
of every ciphertext in every block. This can only make the prover stronger.
Since its response is classical, randomizing the response cannot improve its
success probability. Condition on the public message $y$, all these plaintext
labels, and a fixed response. This fixes its error positions and the values
$R_{i,c}$.

There are $K/2+R_{i,c}$ ciphertexts with plaintext $c$ in block $i$.
Exactly $R_{i,c}$ of them are real. Conditioned on the ciphertexts and their
plaintext labels, the real positions form a uniformly random
$R_{i,c}$-subset of these $K/2+R_{i,c}$ positions. Indeed, all real and test
ciphertexts with plaintext $c$ are independent samples from the same
distribution $\WEEnc(1^n,x_i,c)$, and the final permutation is uniform.

If the prover makes at least one error in block $i$, choose one erroneous
position $u$ according to a fixed rule, and let $c$ be its plaintext. If all
tests in block $i$ are nevertheless answered correctly, then $u$ must be a
real position. Hence,
\begin{align}
&\Pr[
  \text{some real answer in block }i\text{ is wrong}
  \ \wedge\
  \text{all tests in block }i\text{ are correct}
  \mid y,\text{ labels},z_y
]\\
&\qquad\leq
\Pr[u\text{ is real}\mid y,\text{ labels},z_y]\\
&\qquad=
    \frac{R_{i,c}}{K/2+R_{i,c}}
\leq
    \frac{\ell}{K/2+\ell}
\leq
    \frac{2\ell}{K}.
\end{align}
If the prover makes no error in block $i$, the event on the left is empty and
the same bound holds.

Since $z_y$ depends only on the revealed view and not on the hidden origins,
conditioning on it does not change the uniform-subset property. Averaging
over $y$, the plaintext labels, and the response preserves the bound. Finally,
$\mathcal B$ implies that this event occurs in at least one block. A union
bound over $i\in[s]$ gives
\[
    \Pr[\mathcal B]\leq\frac{2s\ell}{K}.
\]
Combining this with \cref{eq:hidden-B} proves
\cref{eq:hidden-upper}.
\qed \end{proof}

To recover the numerical value of the protocol using only the decision
oracle, we query several thresholds. Fix a positive inverse polynomial
$\eta$. For all sufficiently large $n$, assume $\eta\leq1$. Set
\[
    M:=\left\lceil\frac{1}{\eta^2}\right\rceil.
\]
For $\theta\in[0,1]$, let
$(V^{\Gamma,\theta}_1,V^{\Gamma,\theta}_2)$ be the protocol that runs $M$
independent copies of $(V^\Gamma_1,V^\Gamma_2)$ in parallel and accepts if
at least $\theta M$ of them accept.

\begin{lemma}[Value of the repeated protocol]\label{lem:repeat}
For every $x$ and every $\theta\in[0,1]$,
\begin{equation}
    \val_{V^{\Gamma,\theta}_1,V^{\Gamma,\theta}_2}(x)
    =
    \Pr[\Bin(M,v)\geq\theta M],
    \qquad
    v:=\val_{V^\Gamma_1,V^\Gamma_2}(x).
\end{equation}
\end{lemma}

\begin{proof}
For a verifier message $y$, define
\[
    q(y)
    :=
    \max_{z_y}
    \Pr[V^\Gamma_2(x,\rho,y,z_y)=1\mid y].
\]
Then
\[
    v=\mathbb E_y[q(y)].
\]

Now consider the $M$-fold protocol and fix all verifier messages
$y^{(1)},\ldots,y^{(M)}$. Conditioned on these messages, the private verifier
states are independent. For any fixed response vector, the $M$ acceptance
indicators are therefore independent. Their success probabilities depend
coordinatewise on the responses supplied to the corresponding copies.

The prover's response is classical, and its success probability is linear in
any distribution over response vectors. Thus, it suffices to consider a fixed
response vector for every tuple of verifier messages.

The probability that at least $\theta M$ copies accept is nondecreasing in
each of these success probabilities. Hence an optimal prover maximizes every
coordinate separately. It attains success probability $q(y^{(r)})$ in copy
$r$. A response in one copy cannot improve the acceptance probability of a
different copy.

After averaging over the independent verifier messages, the $M$ acceptance
indicators are independent and identically distributed Bernoulli variables.
Each has success probability
\[
    \mathbb E_y[q(y)]=v.
\]
Their sum therefore has distribution $\Bin(M,v)$, proving the claim.
\qed \end{proof}

We can now define the estimator. Let
\[
    J:=\left\lfloor\frac{1}{\eta}\right\rfloor.
\]

\begin{center}
\fbox{\begin{minipage}{0.92\textwidth}
$\widetilde{\Gamma}^{O}(x)$:
\begin{enumerate}
  \item For every $j\in\{0,1,\ldots,J\}$, submit
  \[
      \big(V^{\Gamma,j\eta}_1,V^{\Gamma,j\eta}_2,x\big)
  \]
  to the oracle $O$, and let
  \[
      c_j:=
      O\big(V^{\Gamma,j\eta}_1,V^{\Gamma,j\eta}_2,x\big)
  \]
  be its output.

  \item Set
  \[
      \widehat v:=\eta\max\{j:c_j=1\}.
  \]

  \item Output $1$ with probability $\widehat v$, and output $0$ otherwise.
\end{enumerate}
\end{minipage}}
\end{center}

\begin{lemma}[Accuracy of the estimator]\label{lem:est}
For every $O\in\mathcal{O}_{\QCIP[2]}$ and every $x$,
\begin{equation}
    \left|
    \Pr[\widetilde{\Gamma}^{O}(x)=1]
    -
    \val_{V^\Gamma_1,V^\Gamma_2}(x)
    \right|
    \leq2\eta.
\end{equation}
Moreover, $\widetilde{\Gamma}$ is $\QPT$ and makes $J+1$ classical queries
to $O$.
\end{lemma}

\begin{proof}
Write
\[
    v:=\val_{V^\Gamma_1,V^\Gamma_2}(x).
\]
By \cref{lem:repeat} and Hoeffding's inequality, if
$v\geq\theta+\eta$, then
\[
    \val_{V^{\Gamma,\theta}_1,V^{\Gamma,\theta}_2}(x)
    \geq
    1-e^{-2\eta^2M}
    \geq
    1-e^{-2}
    >
    \frac23.
\]
Similarly, if $v\leq\theta-\eta$, then
\[
    \val_{V^{\Gamma,\theta}_1,V^{\Gamma,\theta}_2}(x)
    \leq
    e^{-2\eta^2M}
    \leq
    e^{-2}
    <
    \frac13.
\]
Therefore, by \cref{def:qcipOracle}, every valid oracle $O$ satisfies
\[
    c_j=1
    \quad\text{whenever}\quad
    j\eta\leq v-\eta,
\]
and
\[
    c_j=0
    \quad\text{whenever}\quad
    j\eta\geq v+\eta.
\]
We do not use the oracle's answers at the remaining thresholds.

The query with $j=0$ always returns $c_0=1$, since the repeated protocol with
threshold $0$ accepts unconditionally. Hence $\widehat v$ is well defined.
Every $j$ with $c_j=1$ satisfies $j\eta<v+\eta$, so
\[
    \widehat v\leq v+\eta.
\]
If $v\geq\eta$, let $j^\star$ be the largest index satisfying
$j^\star\eta\leq v-\eta$. Such an index belongs to
$\{0,\ldots,J\}$. By maximality,
\[
    c_{j^\star}=1
    \qquad\text{and}\qquad
    j^\star\eta>v-2\eta.
\]
If $v<\eta$, then $\widehat v\geq0>v-2\eta$. Thus, in all cases,
\[
    |\widehat v-v|\leq2\eta.
\]
Since the estimator outputs $1$ with probability $\widehat v$, this proves
the first claim.

Finally, $J$ and $M$ are polynomial because $\eta$ is inverse polynomial.
Each queried verifier is $\QPT$, so $\widetilde{\Gamma}$ is $\QPT$ and makes
exactly $J+1$ classical oracle queries.
\qed \end{proof}

\subsection{Proof of \cref{thm:main}}

\begin{proof}
Fix the induced two-stage process $\Gamma=(\Gamma_1,\Gamma_2)$ and set
\begin{equation}\label{eq:parameters}
  K:= 2
  \left\lceil\frac{256s^2\ell}{\epsilon}\right\rceil,
  \qquad
  \gamma:=
  \frac{\epsilon}{512s^2K},
  \qquad
  \eta:=
  \frac{\epsilon}{256s}.
\end{equation}
For all sufficiently large $n$, $K\ge\ell$. The parameters
$K$, $M=\lceil1/\eta^2\rceil$, and $J=\lfloor1/\eta\rfloor$ are polynomial,
while $\gamma$ and $\eta$ are inverse polynomial.

Call an instance $x$ \emph{good} if $\Delta(x)\le\gamma(n)$. By Step~1 of
the proof of \cref{lem:decoff}, every
$x\in L\cap\{0,1\}^n$ is good for all sufficiently large $n$, because
$\nu(n)\le\gamma(n)$ eventually. By \cref{lem:decoff},
\[
    \Pr_{x\leftarrow\bLb_n}[x\text{ is not good}]
    \le\negl(n).
\]
Hence a tuple sampled from either $\bL_n^{\otimes s}$ or
$\bLb_n^{\otimes s}$ has all coordinates good except with negligible
probability.

We use the following coupling observation repeatedly. The yes, no, and ideal
experiments all use the same first-stage process $\Gamma_1$ and then apply
the same second-stage process $\Gamma_2$ to an answer vector. Whenever two experiments produce the same answer vector, the input
to $\Gamma_2$ is identical. Therefore the difference between their output
probabilities is at most the probability that their answer vectors differ.

\smallskip
\noindent\textbf{Claim 1.}
\begin{equation}\label{eq:claim1}
\left|
\Pr[\DecYes^{\SampYes}_{\Gamma,D}(1^n)=1]
-
\Ex_{x\leftarrow\bL_n^{\otimes s}}[\iota(x)]
\right|
\le
s\ell\,\nu(n).
\end{equation}

In the yes experiment, the answer to query $(i,j)$ is
\[
    a_{i,j}
    =
    D^*(x_i,\ct_{i,j},|\psi_i\rangle^{\otimes t}),
\]
while the corresponding ideal answer is $\widetilde m_{i,j}$. At the moment
$\ct_{i,j}$ is generated, conditioned on the preceding transcript and on
$\widetilde m_{i,j}$, it is an honest encryption under $x_i$.
Moreover, $\Gamma_1$ never receives the witness copies. Hence correctness
gives
\[
    \Pr[a_{i,j}\neq\widetilde m_{i,j}]
    \le\nu(n),
\]
even after conditioning on the preceding transcript. A union bound over the
$s\ell$ queries and the coupling observation proves \cref{eq:claim1}.

\smallskip
\noindent\textbf{Claim 2.}
\begin{equation}\label{eq:claim2}
\left|
\Pr[\DecNo^{\SampNo}_{\Gamma,S^*}(1^n)=1]
-
\Ex_{x\leftarrow\bLb_n^{\otimes s}}[\iota(x)]
\right|
\le
2s\ell\gamma+\negl(n).
\end{equation}

By \cref{lem:decoff} and a union bound over $i\in[s]$, all coordinates of
$x\leftarrow\bLb_n^{\otimes s}$ are good except with negligible probability.
Condition on an all-good tuple. At the moment $\ct_{i,j}$ is generated,
conditioned on the preceding transcript and its plaintext
$\widetilde m_{i,j}$, it is an honest encryption under $x_i$.
Therefore, by \cref{lem:permsg},
\[
    \Pr[
    S^*(x_i,\ct_{i,j})\neq\widetilde m_{i,j}
    \mid \text{preceding transcript}
    ]
    \le
    2\Delta(x_i)
    \le
    2\gamma.
\]
A union bound over the $s\ell$ queries and the coupling observation proves
\cref{eq:claim2}.

\smallskip
\noindent\textbf{Claim 3.}
\begin{equation}\label{eq:claim3}
\left|
\Ex_{x\leftarrow\bL_n^{\otimes s}}[\iota(x)]
-
\Ex_{x\leftarrow\bLb_n^{\otimes s}}[\iota(x)]
\right|
\le
\frac{\epsilon}{16}+\negl(n).
\end{equation}

Let $\widetilde{\Gamma}$ be the estimator defined above. On every all-good
tuple, \cref{lem:hidden,lem:est} imply, for every
$O\in\mathcal{O}_{\QCIP[2]}$,
\begin{align}
\left|
\Pr[\widetilde{\Gamma}^{O}(x)=1]-\iota(x)
\right|
&\le
2\eta
+
4sK\gamma
+
\frac{s\ell}{K}
=: \beta.
\label{eq:esterr}
\end{align}
By \cref{eq:parameters},
\[
    2\eta=\frac{\epsilon}{128s},
    \qquad
    4sK\gamma=\frac{\epsilon}{128s},
\]
and
\[
    \frac{s\ell}{K}
    \le
    \frac{\epsilon}{256s}.
\]
Therefore
\[
    \beta
    \le
    \frac{5\epsilon}{256s}
    \le
    \frac{\epsilon}{32s}.
\]

Apply condition~4 of \cref{def:gap} to the $\QPT$ distinguisher
$\widetilde{\Gamma}$. There exists an oracle
$O\in\mathcal{O}_{\QCIP[2]}$ such that
\begin{equation}\label{eq:gap-estimator}
\left|
\Pr_{x\leftarrow\bL_n^{\otimes s}}
[\widetilde{\Gamma}^{O}(x)=1]
-
\Pr_{x\leftarrow\bLb_n^{\otimes s}}
[\widetilde{\Gamma}^{O}(x)=1]
\right|
\le
\negl(n).
\end{equation}
For this same oracle, \cref{eq:esterr} holds on every all-good tuple.
All-good tuples have all but negligible probability under both distributions.
Therefore,
\[
\left|
\Pr_{x\leftarrow\bL_n^{\otimes s}}
[\widetilde{\Gamma}^{O}(x)=1]
-
\Ex_{x\leftarrow\bL_n^{\otimes s}}[\iota(x)]
\right|
\le
\beta+\negl(n),
\]
and the same bound holds with $\bL_n^{\otimes s}$ replaced by
$\bLb_n^{\otimes s}$. Combining these two bounds with
\cref{eq:gap-estimator} gives
\[
\left|
\Ex_{x\leftarrow\bL_n^{\otimes s}}[\iota(x)]
-
\Ex_{x\leftarrow\bLb_n^{\otimes s}}[\iota(x)]
\right|
\le
2\beta+\negl(n)
\le
\frac{\epsilon}{16}+\negl(n),
\]
which proves \cref{eq:claim3}.

\smallskip
Combining \cref{eq:claim1,eq:claim2,eq:claim3} by the triangle inequality,
and using
\[
    2s\ell\gamma
    =
    \frac{\epsilon\ell}{256sK}
    \le
    \frac{\epsilon}{256},
\]
we obtain
\begin{align}
&\left|
\Pr[\DecYes^{\SampYes}_{\Gamma,D^*}(1^n)=1]
-
\Pr[\DecNo^{\SampNo}_{\Gamma,S^*}(1^n)=1]
\right|\\
&\qquad\le
s\ell\,\nu(n)
+
\frac{\epsilon}{256}
+
\frac{\epsilon}{16}
+
\negl(n).
\end{align}
Since $\epsilon$ is inverse polynomial while $s\ell\,\nu(n)$ and the
remaining error are negligible, the right-hand side is at most
$\epsilon(n)$ for all sufficiently large $n$. This proves
\cref{thm:main}.
\qed \end{proof}

\section{Separation}

We now use the results of Section~\ref{sec:main} to show that oracle access to
the inefficient \textsf{WE} adversary can be simulated by an efficient
algorithm.

\begin{theorem}\label{thm:sep}
Let $\epsilon=\epsilon(n)$ be any positive inverse polynomial. Let $L$ be a
$\QMA$ language with verification algorithm $Q$ and a $\Gap$ gap problem, and
let $\Pi=(\WEEnc,\WEDec)$ be a \textsf{WE} scheme for $R_{L,Q}$. Then there exists an algorithm
$\overline{P}$ satisfying the following conditions:
\begin{itemize}
\item $\overline{P}$ is a $\Pi$-adversary with advantage $1/4$.

\item Oracle access to $\overline{P}$ is efficiently simulatable. In
particular, for every falsifiable assumption $(C, c)$ and every
restricted \textsf{WE}--reduction $\Sigma$ showing the security of $\Pi$ from
$(C, c)$, there exists a QPT algorithm $P$ such that
\begin{equation}\label{eq:sim}
\Big|
\Pr[\Gamma^{\overline{P}}(1^n)=1]
-
\Pr[\Gamma^{P}(1^n)=1]
\Big|
\le
\epsilon(n)
\end{equation}
for all sufficiently large $n$, where
$\Gamma^{\overline{P}}(1^n) := \langle C(1^n), \Sigma^{\overline{P}}(1^n)\rangle$
and
$\Gamma^{P}(1^n) := \langle C(1^n), \Sigma^{P}(1^n)\rangle$.
\end{itemize}
\end{theorem}

\begin{proof}
Let $\overline{P}=(\overline{P}_1,\overline{P}_2)$ be the following
possibly inefficient algorithm. On input $1^n$, $\overline{P}_1$ samples
\[
x\leftarrow\SampNo(1^n)
\]
and outputs $x$. On input $(x,\ct)$, $\overline{P}_2$ outputs
\[
m\leftarrow S^*(x,\ct),
\]
where $S^*$ is the classical function of \cref{def:bayes}. Note that
$\overline{P}$ depends neither on $\Sigma$ nor on $\epsilon$.

Now fix any restricted \textsf{WE}--reduction $\Sigma$. Let
$P=(P_1,P_2)$ be the following $\QPT$ algorithm. On each invocation with
input $1^n$, $P_1$ samples
\[
(x,|w_x\rangle^{\otimes t\ell})
\leftarrow
\SampYes(1^n,1^{t\ell}),
\]
outputs $x$, and keeps the witness copies as its private state, where $t$ is the number of witness copies used in decryption. 

On input $(x,\ct)$, $P_2$ computes
\[
m\leftarrow \WEDec(x,\ct,|w_x\rangle^{\otimes t}),
\]
consuming $t$ witness copies for each query.
Since a restricted \textsf{WE}--reduction makes at most $\ell$ ciphertext
queries for each instance, the $t\ell$ witness copies are sufficient.

\smallskip
\noindent\textbf{$\overline{P}$ is a $\Pi$-adversary.}
Consider the security game
$\WESec_{\Pi,\overline{P}}(1^n)$. A bit $b\leftarrow\{0,1\}$ is sampled,
$\overline{P}_1$ samples an instance $x\leftarrow\SampNo(1^n)$, and the challenger computes
\[
\ct\leftarrow\WEEnc(1^n,x,b).
\]
Then $\overline{P}_2$ outputs $b'\leftarrow S^*(x,\ct).$

The game outputs $1$ exactly when $b' = b$, since $\mathsf{SampNo}$
outputs $x \in \overline{L}_n$ with probability $1$. By the definition of $S^*$, this happens with probability $1-\Delta(x)$. Therefore,
\begin{equation}
\Pr\big[\WESec_{\Pi,\overline{P}}(1^n)=1\big]
=
1-
\Ex_{x\leftarrow\SampNo(1^n)}
\big[\Delta(x)\big].
\end{equation}

Apply \cref{lem:decoff} with $\gamma(n):=1/n$. Since
$\Delta(x)\le 1/2$ by \cref{eq:bayes-error-alt},
\begin{align}
\Ex_{x\leftarrow\SampNo(1^n)}[\Delta(x)]
&\le
\frac{1}{n}
+
\Pr_{x\leftarrow\SampNo(1^n)}
\left[\Delta(x)\ge\frac{1}{n}\right]  \\
&\le
\frac{1}{n}+\negl(n).
\end{align}
Hence,
\[
\Pr\big[\WESec_{\Pi,\overline{P}}(1^n)=1\big]
\ge
1-\frac1n-\negl(n)
\ge
\frac34
\]
for all sufficiently large $n$. Thus, $\overline{P}$ is a $\Pi$-adversary
with advantage $1/4$ in the sense of \cref{def:adversary}.

\smallskip
\noindent\textbf{Oracle access is simulatable.}
Fix a $\QPT$ challenger $C$ and reduction $\Sigma$, and let
$\Gamma$ be the composition.

By \cref{lem:decomposition}, the output of the phased execution of
\cref{def:decomposition} with the adversary $\overline{P}$ is
distributed identically to that of $\Gamma^{\overline{P}}(1^n)$. This
phased execution is precisely
$\mathsf{DecNo}^{\mathsf{SampNo}}_{\Gamma, S^*}(1^n)$. Indeed, the first stage of $\overline{P}$ samples its instances using
$\SampNo$, while its second stage answers every ciphertext using $S^*$.

Similarly, $\Gamma^{P}(1^n)$ is distributed identically  to  the yes-based experiment $\DecYes^{\SampYes}_{\Gamma,D}(1^n).$ The first stage of $P$ samples each instance together with sufficiently many
witness copies using $\SampYes$, and its second stage answers each ciphertext
using $\WEDec$ with fresh witness copies.

Therefore, \eqref{eq:sim} is exactly the conclusion of
\cref{thm:main}, and the claim follows.
\qed \end{proof}

\section{Impossibility Result}

We are now ready to prove \cref{thm:informalmain}, showing the impossibility
of restricted quantum black-box reductions from \textsf{WE} to falsifiable
cryptographic assumptions.

\begin{theorem}\label{thm:imposs}
Let $L$ be a language in $\QMA$ with verification algorithm $Q$ and relation
$R_{L,Q}$. Assume that $L$ has a $\Gap$ gap problem. Let
$\Pi=(\WEEnc,\WEDec)$ be a \textsf{WE} scheme for $R_{L,Q}$. Then, for any
falsifiable cryptographic assumption $(C,c)$, one of the following statements
holds:
\begin{enumerate}
  \item The assumption $(C,c)$ is false.
  \item There is no restricted \textsf{WE}--reduction showing the security of
  $\Pi$ from $(C,c)$.
\end{enumerate}
\end{theorem}

\begin{proof}
Suppose that statement~2 is false. Then there exists a restricted
\textsf{WE}--reduction $\Sigma$ establishing the security of $\Pi$ from
$(C,c)$.

Apply \cref{def:reduction} for black-box reduction with the constant polynomial $p\equiv 4$.
There exists a polynomial $p'$ such that, for every $\Pi$-adversary
$\overline{P}$ with advantage $1/4$,
\begin{equation}\label{eq:redwins}
  \Pr\big[
  1\gets\langle C(1^n),\Sigma^{\overline{P}}(1^n)\rangle
  \big]
  \ge
  c+\frac{1}{p'(n)}
\end{equation}
for infinitely many $n\in\mathbb{N}$. Note that $p'$ is fixed independently
of the choice of $\overline{P}$.

Set
\[
  \epsilon(n):=\frac{1}{2p'(n)}.
\]
By \cref{thm:sep}, there exists an inefficient $\Pi$-adversary
$\overline{P}$ with advantage $1/4$ such that oracle access to
$\overline{P}$ can be simulated efficiently. In particular, for the
restricted reduction $\Sigma$ and challenger $C$, there exists a $\QPT$
algorithm $P$ such that
\begin{equation}\label{eq:imposs-sim}
  \Big|
  \Pr[\Gamma^{\overline{P}}(1^n)=1]
  -
  \Pr[\Gamma^{P}(1^n)=1]
  \Big|
  \le
  \epsilon(n)
\end{equation}
for all sufficiently large $n$, where
$\Gamma$ denotes the composed interaction between $C$ and
$\Sigma$.

Hence, combining \eqref{eq:redwins} and
\eqref{eq:imposs-sim}, we obtain, for infinitely many $n$,
\begin{align}
  \Pr\big[
  1\gets\langle C(1^n),\Sigma^{P}(1^n)\rangle
  \big]
  &\ge
  c+\frac{1}{p'(n)}-\epsilon(n)\\
  &=
  c+\frac{1}{2p'(n)}.
\end{align}

Since both $P$ and $\Sigma$ are $\QPT$, the composed algorithm
$\Sigma^P$ is also $\QPT$. Thus, $\Sigma^P$ breaks the falsifiable
assumption $(C,c)$ with non-negligible advantage. Therefore the assumption
is false, proving statement~1.
\qed \end{proof}

\textbf{Acknowledgements.}
TM is supported by
JST CREST JPMJCR23I3,
JST Moonshot R\verb|&|D JPMJMS2061-5-1-1, 
JST FOREST, 
MEXT QLEAP, 
the Grant-in Aid for Transformative Research Areas (A) 21H05183,
and 
the Grant-in-Aid for Scientific Research (A) No.22H00522.

\textbf{AI Disclosure.}
We used 
Claude Fable 5 and ChatGPT 5.6 to assist with writing proofs and texts.
The authors verified the correctness and originality of all content including references.

\newpage

\printbibliography

\newpage
\appendix

\section{Separating Indistinguishability Obfuscation}\label{sec:extensions}

We formalize \cref{cor:informalqio}, ruling out black-box
reductions for indistinguishability obfuscation of null quantum circuits to falsifiable assumptions. This is a direct consequence of \cref{thm:imposs} and \cref{thm:bm}. 

We first record the structural form of the construction of~\cite{BM21}, which builds \textsf{WE} for $\QMA$ from null--\textsf{qIO}.

\begin{proposition}\label{prop:construction}
Let $\mathcal{O}=(\mathsf{Obf},\mathsf{Eval})$ be a null--\textsf{qIO} scheme
(\cref{def:nulliO}) and let $L\in\mathsf{QMA}$ with verification algorithm $Q$ and
relation $R_{L,Q}$. The \textsf{WE} scheme for $L$ $\Pi[\mathcal{O}]=(\mathsf{WE.Enc},
\mathsf{WE.Dec})$ of \cref{thm:bm} has the following form: there is a polynomial
$t=t(n)$ and a deterministic polynomial-time map $x\mapsto(Q^{x}_{0},Q^{x}_{1})$
sending an instance to descriptions of two quantum circuits, such that
\begin{enumerate}
  \item $\mathsf{WE.Enc}(1^n,x,b)=\mathsf{Obf}(1^n,Q^{x}_{b})$ for
    $b\in\{0,1\}$;
  \item $\mathsf{WE.Dec}(x,\mathsf{ct},\rho)=\mathsf{Eval}(\mathsf{ct},\rho)$;
  \item $Q^{x}_{0}$ is the null circuit for every $x$, and $Q^{x}_{1}$ is a null circuit whenever $x\notin L$.
  \item for every $x\in L$ with witness $\lvert w_x\rangle$,
    $\Pr[Q^{x}_{1}(\lvert w_x\rangle^{\otimes t})=1]\ge 1-\mathsf{negl}(n)$.
\end{enumerate}
\end{proposition}

Here $Q^{x}_{1}$ is the soundness-amplified verifier $Q(x,\cdot)$: running $Q$
in parallel on independent witness copies and thresholding drives the $\QMA$
soundness error below any negligible function, so that $Q^{x}_{1}$ rejects all
inputs except with negligible probability when $x\notin L$ (item~3), while still
accepting the multi-copy witness when $x\in L$ (item~4). The circuits are padded to have the same size and input/output lengths. Items~1--4 recover the
correctness (\cref{def:weCorrect}) and security (\cref{def:weSec}) of $\Pi[\mathcal{O}]$:
decryption evaluates the obfuscation on the witness, and for $x\notin L$ both
$Q^{x}_{0},Q^{x}_{1}$ are null, so the two obfuscations are indistinguishable by
null--\textsf{qIO} security.

The security notion for null--\textsf{qIO} in \cite{BM21} (\cref{def:null-qio-security}) is defined relative to a non-interactive non-uniform adversary. We present a weaker definition here mirroring \cref{def:weSec}. As in the \textsf{WE} game, the adversary is permitted to
choose the challenge itself (here, the pair of null circuits) and we only consider uniform adversaries. Both changes strengthen the resulting separation. Note that it is straightforward to show that the \textsf{WE} construction from null--\textsf{qIO} holds relative to this weakened security notion.  

\begin{definition}[Null--\textsf{qIO} security game]\label{def:qio-game}
Let $\mathcal{O}=(\mathsf{Obf},\mathsf{Eval})$ be a null--\textsf{qIO} scheme and
$A=(A_1,A_2)$ be a pair of quantum algorithms.

\medskip
\noindent$\mathsf{qIO.Sec}_{\mathcal{O},A}(1^{n})$:
\begin{enumerate}
  \item $b\leftarrow\{0,1\}$.
  \item $(Q_0,Q_1,z)\leftarrow A_1(1^{n})$, where $z$ is some classical string.
  \item If the circuits $(Q_0,Q_1)$ are invalid (do not have the same size and the same input and output spaces) output $0$. Otherwise, compute $\widetilde Q\leftarrow\mathsf{Obf}(1^{n},Q_b)$.
  \item $b'\leftarrow A_2(Q_0,Q_1,z,\widetilde Q)$.
  \item Output $1$ if $Q_0$ and $Q_1$ are both null circuits and $b'=b$;
    otherwise output $0$.
\end{enumerate}
\end{definition}
\begin{remark}
    Note that in \cref{def:qio-game} the first phase of the adversary $A_1$ can pass a classical string $z$ to the second phase $A_2$. This is necessary to mirror the \textsf{WE} game where the instance $x$ is sent from the first phase to the second phase. Note that \cref{def:qio-game} is still weaker than the standard security notion (\cref{def:null-qio-security}), which allows for quantum advice. 
\end{remark}

\begin{definition}[\textsf{qIO}--adversary]\label{def:qio-adv}
A quantum algorithm $A=(A_1,A_2)$ is a \emph{\textsf{qIO}--adversary} for $\mathcal{O}$ if
there is a polynomial $p$ such that
$\Pr[\mathsf{qIO.Sec}_{\mathcal{O},A}(1^{n})=1]\ge \tfrac12+\tfrac{1}{p(n)}$
for infinitely many $n$.
\end{definition}

\begin{definition}[Null--\textsf{qIO}--reduction]\label{def:qio-red}
Let $(C,c)$ be a falsifiable assumption. A quantum black-box reduction showing the
security of $\mathcal{O}$ from $(C,c)$ is a $\QPT$ algorithm $\Sigma^{(\cdot)}$ such
that for every polynomial $p$ there is a polynomial $p'$ with the following
property: for every \textsf{qIO}--adversary $A$ of advantage $1/p$, the algorithm
$\Sigma^{A}$ breaks $(C,c)$ with advantage $1/p'$ i.e.
\begin{equation}
  \Pr\big[1\leftarrow\langle C(1^n),\Sigma^{A}(1^n)\rangle\big]\ \ge\ c+\frac{1}{p'(n)},
\end{equation}
for infinitely many $n\in\mathbb{N}$. 
For short, we call such a reduction a \emph{\textsf{null--qIO}--reduction}.
\end{definition}

Our separation applies to reductions with the following 
structure.
\begin{definition}[Restricted null--\textsf{qIO} reduction]
\label{def:restricted-qio-reduction}
Let $\mathcal{O}=(\mathsf{Obf},\mathsf{Eval})$ be a null--\textsf{qIO}
scheme, let $(C,c)$ be a falsifiable cryptographic assumption, and let
$\Sigma^{(\cdot)}$ be a null--\textsf{qIO} reduction showing the security
of $\mathcal{O}$ from $(C,c)$.

We say that $\Sigma^{(\cdot)}$ is a \emph{restricted null--\textsf{qIO}
reduction} if there exist polynomials $s=s(n)$ and $\ell=\ell(n)$ such
that, in the interaction
\[
    \left\langle C(1^n),\Sigma^{A}(1^n)\right\rangle,
\]
the oracle access of $\Sigma$ to a possibly inefficient
null--\textsf{qIO} adversary $A=(A_1,A_2)$ has the following form.

\begin{enumerate}
    \item
    The challenger $C$ and the reduction $\Sigma$ may interact
    arbitrarily throughout the execution. Only $\Sigma$ has oracle
    access to $A$.

    \item
    For every $i\in[s]$, the reduction submits $1^n$ to $A_1$ and
    receives a syntactically valid pair of quantum-circuit descriptions
    \[
        (Q_i^0,Q_i^1,z_i)\leftarrow A_1(1^n).
    \]
    where $z_i$ is a polynomial-length classical string and
$Q_i^0,Q_i^1$ are syntactically valid circuit descriptions having the
same size and the same input and output spaces. 

    \item
    For every $i\in[s]$ and $j\in[\ell]$, a bit
    $b_{i,j}\gets\{0,1\}$ is sampled, and
    \[
        \widetilde Q_{i,j}
        \leftarrow
        \mathsf{Obf}\!\left(1^n,Q_i^{b_{i,j}}\right)
    \]
    is generated honestly. This computation may be performed within $\Sigma$, within $C$,
    or jointly through their interaction. 

    \item
    For every $i\in[s]$ and $j\in[\ell]$, the reduction
    submits $(Q_i^0,Q_i^1,z_i,\widetilde Q_{i,j})$,    and receives
    \[
        \widetilde b_{i,j}
        \leftarrow
        A_2\!\left(
            Q_i^0,Q_i^1,z_i,\widetilde Q_{i,j}
        \right).
    \]
    \item The execution ends when $C$ outputs its
    decision bit.
\end{enumerate}

Thus, the oracle queries are parameter-preserving, classical,
non-adaptive, and made only on
honestly generated obfuscations.
\end{definition}

The following directly follows from \cref{prop:construction} and \cref{thm:imposs}.

\begin{theorem}\label{thm:qio-sep}
Let $L$ be a language in $\mathsf{QMA}$ admitting a $\Gap$ gap problem, and let $\mathcal{O}=(\mathsf{Obf},\mathsf{Eval})$ be a null--\textsf{qIO}
scheme. Then for every falsifiable assumption $(C,c)$, one of the following
holds:
\begin{enumerate}
  \item the assumption $(C,c)$ is false; or
  \item there is no restricted null--\textsf{qIO}--reduction showing the security of $\mathcal{O}$ from $(C,c)$.
\end{enumerate}
\end{theorem}

\end{document}